\documentclass[letterpaper, 10 pt, journal, onside, final, twocolumn, nofonttune]{IEEEtran}  % Comment this line out if you need a4paper
\usepackage[top=0.833in, left=0.667in, right=0.667in, bottom=0.597in, includefoot]{geometry}

\usepackage{graphics} % for pdf, bitmapped graphics files
\usepackage{epsfig} % for postscript graphics files
\usepackage[caption=false]{subfig}

\usepackage{mathptmx} % Use Times as default text font, and provide maths support
\usepackage{times} % Select Adobe Times Roman (or equivalent) as default font
\usepackage{color} % Colour control
\usepackage[table,xcdraw]{xcolor} % Driver-independent color extensions
\usepackage{comment} % Selectively include/exclude portions of text
\usepackage[xcolor={dvipdf,gray}]{changes} 
\definechangesauthor[name={Rafael Rodrigues da Silva}, color=red]{a1}
\definechangesauthor[name={Professor Hai Lin}, color=red]{a2}

\usepackage{amsmath,amssymb,amsthm} % American Mathematical Society (AMS) mathematical facilities
\usepackage{leftidx} % Left and right subscripts and superscripts in math mode
\usepackage{bm}  % assumes amsmath package installed - for boldsymbol
\theoremstyle{plain}
\newtheorem{proposition}{Proposition}

\newtheorem{assumption}{Assumption}
\newtheorem{theorem}{Theorem}

\theoremstyle{definition}
\newtheorem{definitionx}{Definition} 
 \newenvironment{definition}
   {\pushQED{\qed}\definitionx}
   {\popQED\enddefinitionx}
\newtheorem{problemx}{Problem}
\newenvironment{problem}
  {\pushQED{\qed}\problemx}
  {\popQED\endproblemx}
\newtheorem{examplex}{Example}
\newenvironment{example}
  {\pushQED{\qed}\examplex}
  {\popQED\endexamplex}
 
\theoremstyle{remark}
\newtheorem{remark}{Remark}
\usepackage{multirow} % Create tabular cells spanning multiple rows

\usepackage{tikz}
\usetikzlibrary{arrows,positioning,chains,fit,shapes,automata} 
\usepackage{pgfplots}
\usepgfplotslibrary{patchplots}
\pgfplotsset{compat=1.14}
\tikzset{
    >=stealth',
    punkt/.style={
           rectangle,
           rounded corners,
           draw=black, very thick,
           text width=12em,
           minimum height=2em,
           text centered},
    pil/.style={
           ->,
           thick,
           shorten <=2pt,
           shorten >=2pt,},
	cross/.style={
    	cross out, 
        draw=black, 
        minimum size=2*(#1-\pgflinewidth), 
        inner sep=0pt, 
        outer sep=0pt},
	cross/.default={1pt}
}    

\usepackage{listings}

\usepackage{algorithm}
\usepackage[noend]{algorithmic}
\usepackage[symbols,toc]{glossaries-extra}

\setabbreviationstyle[acronym]{long-short}
\preto\chapter{\glsresetall}

\newacronym{bsc}{BSC}{Bounded Satisfiability Checking}
\newacronym{cegis}{CEGIS}{Counterexample-Guided Inductive Synthesis}
\newacronym{iis}{IIS}{Irreducibly Inconsistent Set}
\newacronym{ips}{IPS}{Intelligent Physical System}
\newacronym{lp}{LP}{Linear Programming}
\newacronym{ltl}{LTL}{Linear Temporal Logic}
\newacronym{milp}{MILP}{Mixed-Integer Linear Programming}
\newacronym{mpc}{MPC}{Model Predictive Control}
\newacronym{mtl}{MTL}{Metric Temporal Logic}
\newacronym{ompl}{OMPL}{Open Motion Planning Library}
\newacronym{rtl}{RTL}{Temporal Logic over Reals}
\newacronym{smt}{SMT}{Satisfiability Modulo Theories}
\newacronym{stl}{STL}{Signal Temporal Logic}

\glsxtrnewsymbol % requires glossaries-extra.sty 'symbols' option
[description={length of the vector of continuous state variables}]
{statecnb}% label (and sort value)
{\ensuremath{n}}% symbol
\glsxtrnewsymbol % requires glossaries-extra.sty 'symbols' option
[description={length of the vector of continuous input variables}]
{inputcnb}% label (and sort value)
{\ensuremath{m}}% symbol
\glsxtrnewsymbol % requires glossaries-extra.sty 'symbols' option
[description={the sampling period}]
{ts}% label (and sort value)
{\ensuremath{T_s}}% symbol
\glsxtrnewsymbol % requires glossaries-extra.sty 'symbols' option
[description={tolerance for dynamical feasibility}]
{tolfeas}% label (and sort value)
{\ensuremath{\delta}}% symbol

\glsxtrnewsymbol % requires glossaries-extra.sty 'symbols' option
[description={the vector of continuous state variables}]
{statec}% label (and sort value)
{\ensuremath{\boldsymbol{x}}}% symbol
\glsxtrnewsymbol % requires glossaries-extra.sty 'symbols' option
[description={the vector of initial continuous state}]
{statecini}% label (and sort value)
{\ensuremath{\bar{\boldsymbol{x}}}}% symbol
\glsxtrnewsymbol % requires glossaries-extra.sty 'symbols' option
[description={a vector of continuous input variables}]
{inputc}% label (and sort value)
{\ensuremath{\boldsymbol{u}}}% symbol

\glsxtrnewsymbol % requires glossaries-extra.sty 'symbols' option
[description={the domain of continuous state variables}]
{statecdom}% label (and sort value)
{\ensuremath{\mathcal{X}}}% symbol
\glsxtrnewsymbol % requires glossaries-extra.sty 'symbols' option
[description={the domain of continuous input variables}]
{inputcdom}% label (and sort value)
{\ensuremath{\mathcal{U}}}% symbol
\glsxtrnewsymbol % requires glossaries-extra.sty 'symbols' option
[description={a convex constraint}]
{poly}% label (and sort value)
{\ensuremath{\mathcal{P}}}% symbol

\glsxtrnewsymbol % requires glossaries-extra.sty 'symbols' option
[description={a sequence of states and input variables evaluations}]
{run}% label (and sort value)
{\ensuremath{\boldsymbol{\xi}}}% symbol
\glsxtrnewsymbol % requires glossaries-extra.sty 'symbols' option
[description={a time indexed sequence of convex constraints}]
{polyseq}% label (and sort value)
{\ensuremath{\boldsymbol{P}}}% symbol

\glsxtrnewsymbol % requires glossaries-extra.sty 'symbols' option
[description={an \gls*{rtl} path formula}]
{rtlformula}% label (and sort value)
{\ensuremath{\varphi}}% symbol
\glsxtrnewsymbol % requires glossaries-extra.sty 'symbols' option
[description={an \gls*{rtl} state formula}]
{stateformula}% label (and sort value)
{\ensuremath{\phi}}% symbol
\glsxtrnewsymbol % requires glossaries-extra.sty 'symbols' option
[description={a binary logic operator which states that an interpretation satisfies a formula}]
{sat}% label (and sort value)
{\ensuremath{\vDash}}% symbol
\glsxtrnewsymbol % requires glossaries-extra.sty 'symbols' option
[description={a binary logic operator which states that an interpretation does not satisfies a formula}]
{unsat}% label (and sort value)
{\ensuremath{\not\vDash}}% symbol

\glsxtrnewsymbol % requires glossaries-extra.sty 'symbols' option
[description={a linear affine function}]
{stlfunc}% label (and sort value)
{\ensuremath{\mu}}% symbol

\glsxtrnewsymbol % requires glossaries-extra.sty 'symbols' option
[description={the temporal logic operator always}]
{always}% label (and sort value)
{\ensuremath{\square}}% symbol
\glsxtrnewsymbol % requires glossaries-extra.sty 'symbols' option
[description={the temporal logic operator eventually}]
{eventually}% label (and sort value)
{\ensuremath{\Diamond}}% symbol
\glsxtrnewsymbol % requires glossaries-extra.sty 'symbols' option
[description={the temporal logic operator until}]
{until}% label (and sort value)
{\ensuremath{\boldsymbol{U}}}% symbol
\glsxtrnewsymbol % requires glossaries-extra.sty 'symbols' option
[description={the temporal logic operator release}]
{release}% label (and sort value)
{\ensuremath{\boldsymbol{R}}}% symbol

\glsxtrnewsymbol % requires glossaries-extra.sty 'symbols' option
[description={logical tautology}]
{true}% label (and sort value)
{\ensuremath{\top}}% symbol
\glsxtrnewsymbol % requires glossaries-extra.sty 'symbols' option
[description={logical contradiction}]
{false}% label (and sort value)
{\ensuremath{\perp}}% symbol
\glsxtrnewsymbol % requires glossaries-extra.sty 'symbols' option
[description={logical conjunction}]
{and}% label (and sort value)
{\ensuremath{\wedge}}% symbol
\glsxtrnewsymbol % requires glossaries-extra.sty 'symbols' option
[description={logical disjunction}]
{or}% label (and sort value)
{\ensuremath{\vee}}% symbol
\glsxtrnewsymbol % requires glossaries-extra.sty 'symbols' option
[description={logical negation}]
{neg}% label (and sort value)
{\ensuremath{\neg}}% symbol
\glsxtrnewsymbol % requires glossaries-extra.sty 'symbols' option
[description={logical conditional}]
{implies}% label (and sort value)
{\ensuremath{\rightarrow}}% symbol
\glsxtrnewsymbol % requires glossaries-extra.sty 'symbols' option
[description={logical biconditional}]
{iff}% label (and sort value)
{\ensuremath{\leftrightarrow}}% symbol
\glsxtrnewsymbol % requires glossaries-extra.sty 'symbols' option
[description={if}]
{if}% label (and sort value)
{\ensuremath{\leftarrow}}% symbol
\glsxtrnewsymbol % requires glossaries-extra.sty 'symbols' option
[description={a predicate}]
{pred}% label (and sort value)
{\ensuremath{\pi}}% symbol
\glsxtrnewsymbol % requires glossaries-extra.sty 'symbols' option
[description={a set of predicates}]
{predset}% label (and sort value)
{\ensuremath{\Pi}}% symbol

\usepackage{xifthen}% provides \isempty test
\usepackage{xspace}

\newcommand{\idstlpred}[1][]{\ensuremath{\gls*{pred}^{\ifthenelse{\isempty{#1}}{\gls*{stlfunc}}{#1}}}}
\newcommand{\idstlalways}[2]{\ensuremath{\gls*{always}_{#1} #2}}
\newcommand{\idstlevent}[2]{\ensuremath{\gls*{eventually}_{#1} #2}}
\newcommand{\idstluntil}[3]{\ensuremath{#1 \gls*{until}_{#2} #3}}
\newcommand{\idstlrelease}[3]{\ensuremath{#1 \gls*{release}_{#2} #3}}
\newcommand{\encform}[1]{\left|\left(#1\right)\right|}

\makeatletter
\let\NAT@parse\undefined
\makeatother

\usepackage{hyperref}  %hyperref still needs to be put at the end!

\title{ \LARGE \bf
Automatic Trajectory Synthesis for Real-Time Temporal Logic\thanks{This work was supported in part by the National Science Foundation under Grant IIS-1724070, and Grant CNS-1830335, and in part by the Army Research Laboratory under  Grant W911NF-17-1-0072}
}

\author{Rafael Rodrigues da Silva$^{1}$ Vince Kurtz$^1$, and Hai Lin$^{1}$% <-this % stops a space
	\thanks{$^{1}$ All authors are with Department of Electrical Engineering, University of Notre Dame, Notre Dame, IN 46556, USA.
		{\tt\small (rrodri17@nd.edu;~vkurtz@nd.edu;~hlin1@nd.edu)}}
}

\begin{document}

\maketitle
\thispagestyle{empty}
\pagestyle{empty}

\begin{abstract}
Many safety-critical systems must achieve high-level task specifications with guaranteed safety and correctness. Much recent progress towards this goal has been made through controller synthesis from temporal logic specifications. Existing approaches, however, have been limited to relatively short and simple specifications. Furthermore, existing methods either consider some prior discretization of the state-space, deal only with a convex fragment of temporal logic, or are not provably complete. We propose a scalable, provably complete algorithm that synthesizes continuous trajectories to satisfy non-convex \gls*{rtl} specifications. We separate discrete task planning and continuous motion planning on-the-fly and harness highly efficient boolean satisfiability (SAT) and \gls*{lp} solvers to find dynamically feasible trajectories that satisfy non-convex \gls*{rtl} specifications for high dimensional systems. The proposed design algorithms are proven sound and complete, and simulation results demonstrate our approach's scalability.  
\end{abstract}

%%%%%%%%%%%%%%%%%%%%%%%%%%%%%%%%%%%%%%%%%%%%%%%%%%%%%%%%%%%%%%%%%%%%%%%%%%%%%%%%%%%%%%%%%%%%
\section{Introduction}
%%%%%%%%%%%%%%%%%%%%%%%%%%%%%%%%%%%%%%%%%%%%%%%%%%%%%%%%%%%%%%%
Autonomous \glspl*{ips} must be capable of interpreting and automatically achieving high-level task specifications. \textit{Symbolic control} proposes to fulfill this need by automatically designing controllers that satisfy formal logic specifications. Temporal logics such as \acrfull*{rtl} and \acrfull*{stl} can express a wide variety of tasks for \glspl*{ips} \cite{baier2008principles}. Furthermore, temporal logic formulas are close to natural language and can even be interpreted by verbal commands \cite{finucane2010ltlmop}. 

However, today's large-scale \glspl*{ips} present unprecedented scalability challenges for symbolic control techniques, and existing symbolic control algorithms cannot solve many real-world problems. This scalability challenge stems from the need to combine logical constraints (from task specifications) with continuous motion restrictions (from physical dynamics).

Early efforts in symbolic control relied on discrete abstractions of continuous dynamical systems. Much work focused on obtaining an equivalent discrete and finite quotient transition system (see \cite{belta2007symbolic,kloetzer2008fully,gol2013language,plaku2016motion,belta2017formal} and references therein). Given an equivalent transition system, logical constraints can be handled with efficient search techniques in the discrete space. Finding such discrete abstractions is difficult for high-dimensional systems, however, and these approaches are usually limited to systems with less than five continuous state variables \cite{rungger2013specification}. 

Recently, a growing body of work has focused on the synthesis of continuous trajectories from high-level logical specifications directly (see \cite{wolff2016optimal,plaku2016motion,shoukry2018smc} and references therein). One of the significant challenges in this approach is the combination of logical constraints and physical dynamical constraints. Together, these non-convex constraints mean that even determining whether or not a satisfying trajectory exists is a difficult problem. For this reason, existing trajectory synthesis approaches are only provably complete for bounded-time specifications \cite{belta2019formal}. 

Another challenge caused by joint logical/physical constraints is scalability to high-dimensional systems (those with more than ten continuous state variables) \cite{wolff2016optimal}. State-of-the-art solvers based on \gls*{milp} are exponential in the number of logical predicates, severely limiting scalability to complex systems \cite{franzle2007efficient}. Meanwhile, sampling-based and heuristic approaches can achieve impressive results on certain problems, but are not complete and perform poorly on narrow passages \cite{zhang2008efficient}. 

We directly address this issue for \gls*{rtl} specifications using a two-layer control architecture. By separating the non-convex logical specification (discrete task planning layer) from physical system dynamics (continuous motion planning layer) on-the-fly, we can achieve superior scalability and provable completeness unbounded specifications. 

We focus on \gls*{rtl} in particular because \gls*{rtl} describes specifications over continuous variables and is not time-bounded by definition---the very sort of specification that existing symbolic control techniques struggle to handle. Furthermore, \gls*{rtl} is closely related to commonly used temporal logics like Signal Temporal Logic (STL) and Linear Temporal Logic (LTL). The main difference between RTL and these more commonly used logics is that STL and LTL can include time bounds. In fact, much existing work on symbolic control is restricted to bounded-time subsets of STL and LTL, which enables completeness guarantees \cite{raman2014model}. By using RTL, we essentially consider a complementary subset of specifications---those without time bounds. For this reason, extending our results to STL and LTL should be relatively straightforward.

In the discrete planning layer, we use \gls*{bsc} techniques and highly efficient SAT solvers to overcome nonconvexity in the logical specification. Then, in the continuous motion planning layer, finding a corresponding continuous trajectory is as simple as solving a Linear Program (LP). Inspired by the framework of \gls*{cegis}, these two layers work together to ensure completeness: if a continuous trajectory cannot be found for a given discrete plan, a counterexample is generated to guide the discrete planner at the next iteration. 

Our main contribution is a trajectory synthesis method that is provably sound and complete for unbounded real-time temporal logic specifications. We show that this method is scalable to systems with over 10 state variables and complex logical specifications. Simulation results indicate that our approach is over an order of magnitude faster than the current state-of-the-art. 

The rest of the paper is organized as follows. We review related work in Section \ref{sec:rw}. After presenting the necessary preliminaries and a formal problem statement in Section \ref{sec:preliminaries}, we outline our proposed approach in Section \ref{sec:overview}. Sections \ref{sec:dplan} and \ref{sec:fsearch} provide a detailed description of the discrete task planning and continuous motion planning algorithms. Section \ref{sec:idrtl} details how discrete task planning and continuous motion planning work together, including proofs of soundness and completeness. Finally, Section \ref{sec:results} provides simulation results that illustrate the speed and scalability of our approach, and Section \ref{sec:conclusion} concludes the paper. 

%%%%%%%%%%%%%%%%%%%%%%%%%%%%%%%%%%%%%%%%%%%%%%%%%%%%%%%%%%%%%%%%%%%%%%%%%%%%%%%%%%%%%%%%%%%%
\section{Related Work}\label{sec:rw}
%%%%%%%%%%%%%%%%%%%%%%%%%%%%%%%%%%%%%%%%%%%%%%%%%%%%%%%%%%%%%%%
Existing approaches for symbolic control based on trajectory synthesis can be roughly divided into three categories: \gls*{milp} based 
\added{\cite{wolff2016optimal,raman2014model,sadraddini2015robust,saha2016milp,liu2017communication,lindemann2017robust}}, sampling based \cite{he2015towards,plaku2016motion},  and \gls*{smt} based \cite{shoukry2018smc} approaches. 

The basic idea of the \gls*{milp} approach is to rewrite statements with logical expressions into mixed-integer constraints. The addition of auxiliary binary variables to facilitate this rewriting, however, renders the problem intractable for long trajectories. Thus MILP-based approaches such as BluSTL \cite{raman2014model} have focused on \gls*{mpc}, which limits the duration (i.e., number of time indices) of the search.  LTL OPT \cite{plaku2016motion} proposed an alternative encoding to synthesize controllers from a fragment of \gls*{ltl} and \added{from \gls*{mtl} \cite{saha2016milp}} for longer time horizons. However, this approach faces the same limitations and struggles to efficiently handle nonconvex logical constraints with a long duration (greater than $100$ time indices). 

When considering only a convex fragment of \gls*{stl}, the problem can be efficiently encoded as a \acrfull*{lp} problem \cite{raman2014model} instead. Furthermore, the satisfaction of such an \gls*{stl} formula can be measured using robust semantics, which allows for efficient controller synthesis using control barrier functions \cite{lindemann2019control} and prescribed performance control \cite{lindemann2017prescribed}. However, this convex fragment of \gls*{stl} cannot describe many tasks required by real-world \gls*{ips}. For instance, a quadrotor performing an automated inspection task needs to return to a charging station in a reasonable time infinitely often: this requires nested existential quantifications (always-eventually) and thus cannot be described by a convex fragment of \gls*{stl}. Similarly, a robot operating in a warehouse might have a task that requires logical disjunction: pick up one box OR another box before moving to a goal destination.

Sampling-based approaches such as the \Gls*{ompl} \gls*{ltl} planner \cite{sucan2012the-open-motion-planning-library} propose to combine sampling-based motion planning with discrete search algorithms. Sampling-based motion planning methods are relatively easy to implement and can provide fast solutions to some difficult problems. However, such approaches are suboptimal and are not guaranteed to find a solution if one exists, a property referred to as \textit{(in)completeness}. Instead, they ensure weaker notions of \textit{asymptotical optimality} \cite{karaman2011sampling} and \textit{probabilistic completeness} \cite{hsu2006probabilistic}, meaning that an optimal solution will be provided, if one exists, given sufficient runtime of the algorithm. These difficulties are exemplified in poor performance on problems with narrow passages \cite{zhang2008efficient}. 

Our proposed approach builds on \gls*{smt} based symbolic control, which has been used to generate dynamically-feasible trajectories for \gls*{ltl} specifications \cite{shoukry2018smc}. Modern \gls*{smt} solvers can efficiently find satisfying valuations of extensive formulas with complex Boolean structures combining various decidable theories such as lists, arrays, bit vectors, linear integer arithmetic, and linear real arithmetic \cite{barrett2009satisfiability}. \gls*{smt} based symbolic control from \gls*{ltl} specifications showed encouraging performance for motion planning problems. However, existing approaches are not provably complete.  Furthermore, the implementation of real-time specifications is difficult, as existing methods do not offer explicit bounds regarding when the synthesis algorithm will terminate. 

To the best of our knowledge, our approach is the first trajectory-based symbolic controller for real-time temporal logic that is provably sound and complete for nonconvex unbounded specifications. Furthermore, we present comparative results showing that our approach scales well to a long duration (greater than $100$ time indices) tasks and high-dimensional (greater than $10$ continuous state variables) system dynamics.

%%%%%%%%%%%%%%%%%%%%%%%%%%%%%%%%%%%%%%%%%%%%%%%%%%%%%%%%%%%%%%%%%%%%%%%%%%%%%%%%%%%%%%%%%%%%%%%%%%%%%%%%%%%%%%%%%%%%%%%%%%%%%%%%%%%%%%%%%%%%%%%%%%%%%%%%%%%%%%%%%%%%%%%%%%
\section{Preliminaries}\label{sec:preliminaries}
%%%%%%%%%%%%%%%%%%%%%%%%%%%%%%%%%%%%%%%%%%%%%%%%%%%%%%%%%%%%%%%%%%%%%%%%%%%%%%%%%%%%%%%%%%%%%%%%%%%%%%%%%%%%%%%%%%%%%%%%%%%%%%%%%%%%%%%%%%%%%%%%%%%%%%%%%%%%%%%%%%%%%%%%%%
%%%%%%%%%%%%%%%%%%%%%%%%%%%%%%%%%%%%%%%%%%%%%%%%%%%%%%%%%%%%%%%%%%%%%%%%%%%%%%%%%%%%%%%%%%%%%%%%%%%%%%%%%%%%%%%%%%%%%%%%%%%%%%%%%%%%%%%%%%%%%%%%%%%%%%%%%%%%%%%%%%%%%%%%%%
\subsection{System}\label{sec:system}
%%%%%%%%%%%%%%%%%%%%%%%%%%%%%%%%%%%%%%%%%%%%%%%%%%%%%%%%%%%%%%%%%%%%%%%%%%%%%%%%%%%%%%%%%%%%%%%%%%%%%%%%%%%%%%%%%%%%%%%%%%%%%%%%%%%%%%%%%%%%%%%%%%%%%%%%%%%%%%%%%%%%%%%%%%

Consider the following discrete-time linear control system:
\begin{equation}\label{eq:system}
    \gls*{statec}_{k+1} = A\gls*{statec}_{k} + B\gls*{inputc}_{k}, 
\end{equation}
where $\gls*{statec}_k \in \gls*{statecdom} \subset \mathbb{R}^{\gls*{statecnb}}$ are the state variables, $\gls*{inputc}_k \in \gls*{inputcdom} \subset \mathbb{R}^{\gls*{inputcnb}}$ are the control inputs, $\gls*{statecdom} := \{ \gls*{statec}_k \in \mathbb{R}^{\gls*{statecnb}} | A_{\gls*{statecdom}} \gls*{statec}_k \leq \boldsymbol{b}_{\gls*{statecdom}}  \}$ and $\gls*{inputcdom} := \{ \gls*{inputc}_k \in \mathbb{R}^{\gls*{inputcnb}} | A_{\gls*{inputcdom}} \gls*{inputc}_k \leq \boldsymbol{b}_{\gls*{inputcdom}}  \}$ are full dimensional polytopes,  $A$, $B$, $A_{\gls*{statecdom}}$, $A_{\gls*{inputcdom}}$ are matrices and $\boldsymbol{b}_{\gls*{statecdom}}$, $\boldsymbol{b}_{\gls*{inputcdom}}$ are vectors with proper dimensions. We assume that the system is stable (i.e., $A$ has positive real-valued eigenvalues).

In fact, System (\ref{eq:system}) can arise from linearization and sampling of a more general continuous system. In this case, we denote the sampling period as $\gls*{ts} = t_{k+1} - t_k $.

We will use model checking techniques to verify and control the dynamic behavior of the system. For this purpose, we model this system as a Kripke structure $M_c$. Kripke structures are a type of transition system model that can represent a large class of systems. They are formally defined as follows: 
\begin{definition}\label{def:model}
    A tuple $M = \langle S, Act, T, I, L, \Sigma, F \rangle$ is a Kripke structure where $S$ is a set of states, $I \subseteq S$ is a set of initial states, $Act$ is a set of actions (aka, inputs), $T \subseteq S \times Act \times S$ is a transition relation, $L : S \mapsto 2^{\Sigma}$ is a labeling function over a finite set of symbols $\Sigma$, and $F \subseteq S$ is a set of accepted states.
\end{definition}

A run of $M$ is a sequence $\gls*{run} = s_0 \xrightarrow{\alpha_0} s_1 \xrightarrow{\alpha_1} s_2 \dots s_K$, where $s_k \in S$, $\alpha_k \in Act$, $s_0 \in I$, $s_K \in F$, and $s_k \xrightarrow{\alpha_k} s_{k+1}$ if and only if $(s_k,\alpha_k,s_{k+1}) \in T$. This run generates a path of $M$ which is defined as a sequence of labels $\sigma = L(s_0)L(s_1)L(s_2)\dots L(s_K)$. When we need to model infinite behaviors of $M$, we accept a run if it visits the acceptance set $F$ infinitely often. We call such structures that model infinite behavior \textit{fair} Kripke structures.     

We define a continuous structure $M_c$ for the continuous system (\ref{eq:system}): $M_c = \langle S_c, Act_c, T_c, I_c, L_c, \Sigma, F_c \rangle$, where: $S_c = \gls*{statecdom}$, $I_{c} = \{ \gls*{statecini} \}$ with $\gls*{statecini} \in \gls*{statecdom}$, $Act_c = \gls*{inputcdom}$, $\gls*{statec}_k \xrightarrow{\gls*{inputc}_k} \gls*{statec}_{k+1}$ if and only if $\gls*{statec}_{k+1} = A \gls*{statec}_k + B \gls*{inputc}_k$, and $F_c = \gls*{statecdom}$. The labels $L_c$ and symbols $\Sigma$ depend on the logical specification, and are defined in the next subsection. Therefore, a \textit{run} (trajectory) of $M_c$ is a sequence $\gls*{run}_c = \gls*{statec}_{0}\xrightarrow{\gls*{inputc}_0}\gls*{statec}_{1}\xrightarrow{\gls*{inputc}_1}\dots$.

%%%%%%%%%%%%%%%%%%%%%%%%%%%%%%%%%%%%%%%%%%%%%%%%%%%%%%%%%%%%%%%%%%%%%%%%%%%%%%%%%%%%%%%%%%%%%%%%%%%%%%%%%%%%%%%%%%%%%%%%%%%%%%%%%%%%%%%%%%%%%%%%%%%%%%%%%%%%%%%%%%%%%%%%%%
\subsection{Linear Temporal Logic over Reals}\label{sec:rtl} 
%%%%%%%%%%%%%%%%%%%%%%%%%%%%%%%%%%%%%%%%%%%%%%%%%%%%%%%%%%%%%%%%%%%%%%%%%%%%%%%%%%%%%%%%%%%%%%%%%%%%%%%%%%%%%%%%%%%%%%%%%%%%%%%%%%%%%%%%%%%%%%%%%%%%%%%%%%%%%%%%%%%%%%%%%%

Instead of abstracting the continuous state space directly into finite symbols to provide labels for $M_c$, we use a formal language that allows us to represent state constraints as half-spaces. As we will see in Section \ref{sec:dplan}, this allows us to separate logical constraints from continuous dynamics without the expensive computation of a discrete abstraction. Specifically, we consider high-level specifications are given as \gls*{rtl} formulas \cite{lin2014mission}.

\begin{definition}
    \gls*{rtl} formulas are defined recursively according to the following syntax:
\begin{align*}	
	\gls*{stateformula} := & \idstlpred | \gls*{neg} \idstlpred  | \gls*{stateformula}_1 \gls*{and} \gls*{stateformula}_2 | \gls*{stateformula}_1 \gls*{or} \gls*{stateformula}_2 && \triangleright \text{ state} \\
    \gls*{rtlformula} := & \gls*{stateformula} | \gls*{neg} \gls*{stateformula} | \gls*{rtlformula}_1 \gls*{and} \gls*{rtlformula}_2 | \gls*{rtlformula}_1 \gls*{or} \gls*{rtlformula}_2 | \idstluntil{\gls*{rtlformula}_1}{}{\gls*{rtlformula}_2} | \idstlrelease{\gls*{rtlformula}_1}{}{\gls*{rtlformula}_2}, && \triangleright \text{ RTL}
\end{align*}
where \gls*{rtlformula}, $\gls*{rtlformula}_1$, $\gls*{rtlformula}_2$  are \gls*{rtl} formulas, $\gls*{stateformula}$, $\gls*{stateformula}_1$ and $\gls*{stateformula}_2$ are state formulas, and $\idstlpred \in \gls*{predset}$ is an atomic proposition. Propositions $\idstlpred : \mathbb{R}^{\gls*{statecnb}} \mapsto \{\gls*{true}, \gls*{false}\}$ are defined by a function $\gls*{stlfunc}$, which we assume is linear affine, i.e.,  $\gls*{stlfunc}(\gls*{statec}) = \boldsymbol{h}^\intercal \gls*{statec} + a$, $\boldsymbol{h} \in \mathbb{R}^{\gls*{statecnb}}$ and $a \in \mathbb{R}$.  
\end{definition}

An RTL formula \gls*{rtlformula} is defined in terms of state formulas \gls*{stateformula}. Note that all state formulas are RTL formulas, but not all RTL formulas are state formulas. We denote a state $\gls*{statec}_k \in \mathbb{R}^{\gls*{statecnb}}$ satisfying a state formula by $\gls*{statec} \gls*{sat} \gls*{stateformula}$, and define the notion of satisfaction recursively: $\gls*{statec} \gls*{sat} \idstlpred$ if and only if $\gls*{stlfunc}(\gls*{statec}) > 0$,  $\gls*{statec} \gls*{sat} \gls*{neg} \idstlpred$ if and only if $-\gls*{stlfunc}(\gls*{statec}) > 0$, $\gls*{statec} \gls*{sat} \gls*{stateformula}_1 \gls*{and} \gls*{stateformula}_2$ if and only if $\gls*{statec} \gls*{sat} \gls*{stateformula}_1$ and $\gls*{statec} \gls*{sat} \gls*{stateformula}_2$, and $\gls*{statec} \gls*{sat} \gls*{stateformula}_1 \gls*{or} \gls*{stateformula}_2$ if and only if $\gls*{statec} \gls*{sat} \gls*{stateformula}_1$ or $\gls*{statec} \gls*{sat} \gls*{stateformula}_2$. With these definitions, we can derive standard Boolean shorthands like negation $\gls*{neg}$, implication \gls*{implies}, and biconditional \gls*{iff}.

These state formulas allow us to define symbols $ \Sigma $, and the labeling function $L_c$ for the transition system is $M_c$ as follows:
\begin{definition}
    For each state formula $\gls*{stateformula}$ in the RTL formula $\gls*{rtlformula}$, define the symbol $p^{\gls*{stateformula}}$. Then $\Sigma = \{p^{\gls*{stateformula}}\}$ is the set of all such symbols. The labeling function $L_c : \mathbb{R}^{\gls*{statecnb}} \mapsto 2^{\Sigma}$ maps states to a set of symbols, where $p^{\gls*{stateformula}} \in L_c(\gls*{statec})$ for a state $\gls*{statec} \in \mathbb{R}^{\gls*{statecnb}}$ if and only if $\gls*{statec} \gls*{sat} \gls*{stateformula}$.
\end{definition}

% We represent these state formulas as symbols $p^{\gls*{stateformula}} \in \Sigma$ and denote a labeling function $L_c : \mathbb{R}^{\gls*{statecnb}} \mapsto 2^{\Sigma}$ such that $p^{\gls*{stateformula}} \in L_c(\gls*{statec})$ for a state $\gls*{statec} \in \mathbb{R}^{\gls*{statecnb}}$ of $M_{c}$ if and only if $\gls*{statec} \gls*{sat} \gls*{stateformula}$.   

In this work, we assume that a valid transition of $M_c$ changes satisfaction of at most one predicate. In particular, transitions only occur between adjacent labeled state spaces. This assumption is minimally restrictive, since the system $M_c$ approximate the behavior of a continuous-time system, and regions in state space determine predicates.
\begin{assumption}\label{a:adjacent}
    We assume that a valid transition $(\gls*{statec}_k,\gls*{inputc}_k,\gls*{statec}_{k+1}) \in T_c$ of $M_c$ occurs only if there exists at most one predicate $\idstlpred$ in the specification \gls*{rtlformula} such that $\gls*{statec}_k \gls*{sat} \idstlpred$ and $\gls*{statec}_{k+1} \gls*{sat} \gls*{neg} \idstlpred$.
\end{assumption}

\begin{example}
Consider an integrator $x_{k+1} = x_k + T_s u_k$, where $u_k$, the input variable, is bounded $|u| \leq 1 m/s$, and  $T_s$ is the sampling time. Consider the formula $(x \leq 1) \gls*{or} (x \geq 2)$. The assumption is satisfied if $Ts < 1s$, since at each timestep, at most one of the predicates ($x \leq 1$ or $x \geq 2$) can change. If the sampling time is too large, however, the state could jump from $(x\leq1)$ to $(x\geq2)$ in a single step, violating the assumption.
\end{example}

The meaning (semantics) of an \gls*{rtl} formula is interpreted over a run \gls*{run} of $M_c$.  We denote a run $\gls*{run}$ satisfying an \gls*{rtl} formula \gls*{rtlformula} by $\gls*{run} \gls*{sat} \gls*{rtlformula}$.  
We write $\gls*{run} \gls*{sat}_k \gls*{rtlformula}$ when the run $\gls*{statec}_{k}\xrightarrow{\gls*{inputc}_k}\gls*{statec}_{k+1}\xrightarrow{\gls*{inputc}_{k+1}}\dots$ satisfies the \gls*{rtl} formula \gls*{rtlformula}. 

\begin{definition}
 The following semantics define the validity of a formula \gls*{rtlformula} with respect to the run $\gls*{run}$: 
\begin{itemize}

    \item {$\gls*{run} \gls*{sat} \gls*{rtlformula}$ if and only if $\gls*{run} \gls*{sat}_0 \gls*{rtlformula}$}
    \item {$\gls*{run} \gls*{sat}_{k} \gls*{stateformula}$ if and only if $\gls*{statec}_k \gls*{sat} \gls*{stateformula}$,}
    \item {$\gls*{run} \gls*{sat}_{k} \gls*{rtlformula}_1 \gls*{and} \gls*{rtlformula}_2$ if and only if $\gls*{run} \gls*{sat}_{k} \gls*{rtlformula}_1$ and $\gls*{run} \gls*{sat}_{k} \gls*{rtlformula}_2$,}
    \item {$\gls*{run} \gls*{sat}_{k} \gls*{rtlformula}_1 \gls*{or} \gls*{rtlformula}_2$ if and only if $\gls*{run} \gls*{sat}_{k} \gls*{rtlformula}_1$ or $\gls*{run} \gls*{sat}_{k} \gls*{rtlformula}_2$,}
    \item \added{ {$\gls*{run} \gls*{sat}_{k} \gls*{rtlformula}_1 \gls*{until} \gls*{rtlformula}_2$ if and only if $\exists t_{k^\prime} \geq t_k$ s.t. $\gls*{run} \gls*{sat}_{{k^\prime}}\gls*{rtlformula}_2$, and $\forall t_k \leq t_{k^{\prime\prime}} < t_{k^\prime}, \gls*{run} \gls*{sat}_{{k^{\prime\prime}}}\gls*{rtlformula}_1$,}}
    \item \added{{$\gls*{run} \gls*{sat}_{k} \gls*{rtlformula}_1 \gls*{release} \gls*{rtlformula}_2$ if and only if $\exists t_{k^\prime} \geq t_k$ s.t. $\gls*{run} \gls*{sat}_{{k^\prime}}\gls*{rtlformula}_1$, and $\gls*{run} \gls*{sat}_{{k^{\prime\prime}}}\gls*{rtlformula}_2, \forall t_k \leq t_{k^{\prime\prime}} < t_{k^\prime}$, or $\gls*{run} \gls*{sat}_{{k^\prime}}\gls*{rtlformula}_2, \forall t_{k^\prime} \geq t_k$.}}
\end{itemize}    
\end{definition}

 The operator until $\gls*{rtlformula}_1 \gls*{until} \gls*{rtlformula}_2$ means that the sub-formula $\gls*{rtlformula}_1$ must remain true until $\gls*{rtlformula}_2$ becomes true. On the other hand, a specification  $\gls*{rtlformula}_1$ releases $\gls*{rtlformula}_2$ ($\gls*{rtlformula}_1 \gls*{release} \gls*{rtlformula}_2$) means that $\gls*{rtlformula}_2$ must remain true until $\gls*{rtlformula}_1$ is true. If $\gls*{rtlformula}_1$ is never true, $\gls*{rtlformula}_2$ must remain true forever. Moreover, these definitions allow us to derive the operators ``eventually'' $\idstlevent{}{\gls*{rtlformula}} = \idstluntil{\gls*{true}}{}{\gls*{rtlformula}}$ and ``always'' $\idstlalways{}{\gls*{rtlformula}} = \idstlrelease{\gls*{false}}{}{\gls*{rtlformula}}$.

%The \gls*{rtl} semantics are very similar to \gls*{ltl} semantics without the next operator. The only difference is that the RTL until and release operators are valid if there exists an instant $t_{k^\prime}$ such that $\gls*{run} \gls*{sat}_{{k^\prime}}\gls*{rtlformula}_1 \gls*{and} \gls*{rtlformula}_2$. 

\begin{definition}
We define the set of subformulas (closure) $cl(\gls*{rtlformula})$ of an  \gls*{rtl} formula \gls*{rtlformula} as the smallest set satisfying the following conditions: $\gls*{rtlformula} \in cl(\gls*{rtlformula})$, if $\circ \gls*{rtlformula}_1 \in cl(\gls*{rtlformula})$ for $\circ \in \{ \gls*{eventually}_{}, \gls*{always}_{} \}$ then $\gls*{rtlformula}_1 \in cl(\gls*{rtlformula})$, and if $\gls*{rtlformula}_1 \circ \gls*{rtlformula}_2 \in cl(\gls*{rtlformula})$ for $\circ \in \{ \gls*{or}, \gls*{and}, \gls*{until}, \gls*{release} \}$ then $\gls*{rtlformula}_1,\gls*{rtlformula}_2 \in cl(\gls*{rtlformula})$.    
\end{definition}

\begin{remark}
Note that \gls*{rtl} formulas are time-unbounded by definition. Since dealing with unbounded formulas is particularly difficult for existing symbolic control methods, using RTL rather than related formal logic like STL or MTL allows us to focus on the challenges particular to time-unbounded formulas.
\end{remark}

%%%%%%%%%%%%%%%%%%%%%%%%%%%%%%%%%%%%%%%%%%%%%%%%%%%%%%%%%%%%%%%%%%%%%%%%%%%%%%%%%%%%%%%%%%%%%%%%%%%%%%%%%%%%%%%%%%%%%%%%%%%%%%%%%%%%%%%%%%%%%%%%%%%%%%%%%%%%%%%%%%%%%%%%%%
\subsection{Problem formulation}\label{sec:problem}
%%%%%%%%%%%%%%%%%%%%%%%%%%%%%%%%%%%%%%%%%%%%%%%%%%%%%%%%%%%%%%%%%%%%%%%%%%%%%%%%%%%%%%%%%%%%%%%%%%%%%%%%%%%%%%%%%%%%%%%%%%%%%%%%%%%%%%%%%%%%%%%%%%%%%%%%%%%%%%%%%%%%%%%%%%

The \gls*{rtl} symbolic control problem is formally defined as follows:

\begin{problem}\label{prob:1}
    Given an \gls*{rtl} formula \gls*{rtlformula}, and a dynamical system $M_{c}$, design a control signal $\gls*{inputc} = \gls*{inputc}_{0}\gls*{inputc}_{1}\dots$ such that the corresponding run %generated a dynamically-feasible trajectory
    $\gls*{run} := \gls*{statec}_{0}\xrightarrow{\gls*{inputc}_0}\gls*{statec}_{1}\xrightarrow{\gls*{inputc}_1}\dots$ of the system $M_{c}$ satisfies \gls*{rtlformula}. 
\end{problem}

%%%%%%%%%%%%%%%%%%%%%%%%%%%%%%%%%%%%%%%%%%%%%%%%%%%%%%%%%%%%%%%%%%%%%%%%%%%%%%%%%%%%%%%%%%%%%%%%%%%%%%%%%%%%%%%%%%%%%%%%%%%%%%%%%%%%%%%%%%%%%%%%%%%%%%%%%%%%%%%%%%%%%%%%%%
\section{Overview}\label{sec:overview}
%%%%%%%%%%%%%%%%%%%%%%%%%%%%%%%%%%%%%%%%%%%%%%%%%%%%%%%%%%%%%%%%%%%%%%%%%%%%%%%%%%%%%%%%%%%%%%%%%%%%%%%%%%%%%%%%%%%%%%%%%%%%%%%%%%%%%%%%%%%%%%%%%%%%%%%%%%%%%%%%%%%%%%%%%%

Solving Problem \ref{prob:1} directly in terms of the continuous system $M_c$ is quite difficult due largely to the nonconvexity introduced by the logical specification $\gls*{rtlformula}$. To overcome this nonconvexity, we separate the problem into two parts: discrete task planning and continuous motion planning, as shown in Fig.~\ref{fig:diag1}. In the discrete planning phase, we determine a sequence of convex regions in the state space that enforces satisfaction of the task specification $\gls*{rtlformula}$. Given a discrete plan, finding a corresponding continuous trajectory (motion planning) can be reduced to a simple linear programming problem. 

To find a discrete task plan, we first propose a finite-state abstraction $M_d$ which is related to $M_c$ through a simulation relation \cite{baier2008principles}. Unlike early work in symbolic control, this abstraction is built on-the-fly from the predicates of specification $\gls*{rtlformula}$. Furthermore, we consider discrete plans to be a fair Kripke structure, which allows us to consider unbounded specifications elegantly. Finally, we propose an encoding that allows us to find a satisfying discrete plan by solving a Boolean satisfiability problem (SAT). While SAT is an NP-complete problem, many fast solvers exist,  and SAT/SMT solver performance has been increasing steadily in recent years \cite{franzle2007efficient}.  

Given a discrete plan, we show that finding a corresponding continuous run of a system (\ref{eq:system}) can be reduced to solving a linear program (LP). If this LP is infeasible, we treat the corresponding discrete plan as an infeasible \textit{counterexample}, which is passed back to the discrete planning layer.

A key insight is that the information from previous infeasible discrete plans can be used to generate new plans. Specifically, we show how off-the-shelf incremental \gls*{smt} solvers like Z3 \cite{de2008z3} can use such information from past iterations to improve scalability drastically. 

We prove that our approach is sound (any run generated by our approach satisfies the specification $\gls*{rtlformula}$) and complete (if any satisfying run exists, our approach will find a satisfying run). Furthermore, we demonstrate the scalability of our approach in several simulation examples. Unlike approaches that seek to obtain an equivalent discrete and finite abstraction, our approach obtains a simulation abstraction and does not require feedback controllers that guarantee the transitions. Unlike other trajectory-based approaches, our method guarantees soundness and completeness for unbounded-time specifications. Furthermore, our approach is scalable to high-dimensional systems and complex specifications.

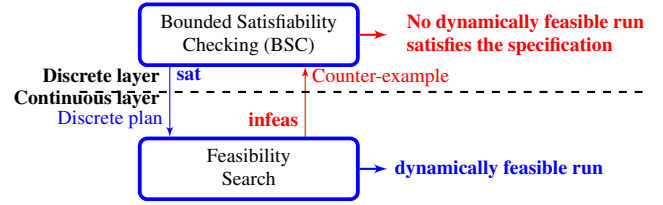
\begin{figure}[htp]
    \tikzstyle{block} = [draw, rounded corners=1mm, color=blue, text=black, line width=0.5mm, rectangle, minimum height=3em, minimum width=6em]
    \centering
    \begin{tikzpicture}[auto, >=latex', scale=0.75, transform shape]
        \node[block, minimum width=11em] (dplan) {\begin{tabular}{c}Bounded Satisfiability \\ Checking (BSC) \end{tabular}};
        \node[block, minimum width=11em,below=1.25cm of dplan] (reachsearch) {\begin{tabular}{c}Feasibility \\ Search\end{tabular} };
        \node[below left = -1mm and -0.5cm of dplan] {\textbf{Discrete layer}};
        \node[below left =  3mm and -0.5cm of dplan] {\textbf{Continuous layer}};
        \draw[->,color=blue] ([xshift=-14mm] dplan.south) -- node[left,near end] {\color{blue} Discrete plan} node[right,pos=0.1] {\color{blue} \textbf{sat}} ([xshift=-14mm] reachsearch.north);
        \draw[->,color=red] ([xshift=10mm]reachsearch.north) -- node[right,pos=0.85] {\color{red} Counter-example} node[left,near start] {\color{red} \textbf{infeas}} ([xshift=10mm] dplan.south);
        \draw[->,thick,color=blue] (reachsearch.east) -- node[right, pos=1] {\color{blue} \textbf{dynamically feasible run}} ++(0.5cm,0);
        \draw[->,thick,color=red] (dplan.east) -- node[right,pos=1] {\color{red} \begin{tabular}{l} \textbf{No dynamically feasible run} \\ \textbf{satisfies the specification} \end{tabular}} ++(0.5cm,0);
        \draw[-,thick,dashed] ([xshift=-3cm,yshift=-4.5mm] dplan.south) -- +(10cm,0);
    \end{tikzpicture}
    \caption{Pictorial representation of proposed approach. }
    \label{fig:diag1}
\end{figure}

%%%%%%%%%%%%%%%%%%%%%%%%%%%%%%%%%%%%%%%%%%%%%%%%%%%%%%%%%%%%%%%%%%%%%%%%%%%%%%%%%%%%%%%%%%%%%%%%%%%%%%%%%%%%%%%%%%%%%%%%%%%%%%%%%%%%%%%%%%%%%%%%%%%%%%%%%%%%%%%%%%%%%%%%%%
\section{Discrete Task Planning}\label{sec:dplan}
%%%%%%%%%%%%%%%%%%%%%%%%%%%%%%%%%%%%%%%%%%%%%%%%%%%%%%%%%%%%%%%%%%%%%%%%%%%%%%%%%%%%%%%%%%%%%%%%%%%%%%%%%%%%%%%%%%%%%%%%%%%%%%%%%%%%%%%%%%%%%%%%%%%%%%%%%%%%%%%%%%%%%%%%%%

In the discrete task planning layer, we generate a sequence of convex constraints that ensures the satisfaction of the specification $\gls*{rtlformula}$. To generate such constraints, we propose a finite discrete transition system $M_d$, which abstracts the behavior of the continuous system $M_c$ with respect to the specification. This discrete system $M_d$, the logical specification $\gls*{rtlformula}$, and any counterexamples can be encoded as Boolean formulas and leverage incremental SAT/SMT solvers to rapidly find a discrete plan, which we represent as a Kripke structure $M_p$. 

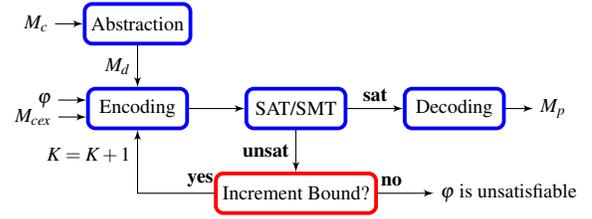
\begin{figure}[htp]
    \tikzstyle{block} = [draw, rounded corners=1mm, color=blue, text=black, line width=0.5mm, rectangle, minimum height=2em, minimum width=5em]
    \tikzstyle{block2} = [draw, rounded corners=1mm, color=red, text=black, line width=0.5mm, rectangle, minimum height=2em, minimum width=5em]
    \centering
    \begin{tikzpicture}[auto, >=latex', scale=0.75, transform shape]	
        \node[block] (abs) {Abstraction};
        \node[block,below=0.75cm of abs] (enc) {Encoding};
        \node[block,right=1cm of enc] (smt) {SAT/SMT};
        \node[block2,below=0.75cm of smt] (cond) {Increment Bound?};
        \node[block,right=1cm of smt] (dec) {Decoding};
        \draw[<-] (abs.west) -- node[left,pos=1] {$M_c$} ++(-0.5cm,0);
        \draw[<-] ([yshift= 1.5mm] enc.west) -- node[left,pos=1] {$\gls*{rtlformula}$} ++(-0.5cm,0);
        \draw[<-] ([yshift=-1.5mm] enc.west) -- node[left,pos=1] {$M_{cex}$} ++(-0.5cm,0);
        \draw[->] (cond.west) -- node[above,midway] {\textbf{yes}} ++(-0.5cm,0) -| node[left,pos=0.8] {$K = K+1$} (enc.south);
        \draw[->] (cond.east) -- node[above,midway] {\textbf{no}} ++(0.5cm,0) --node[right,pos=1] {\gls*{rtlformula} is unsatisfiable} ++(0.5cm,0);
        \draw[->] (abs.south) -- node[left,midway] {$M_d$} (enc.north);
        \draw[->] (enc) -- (smt);
        \draw[->] (smt) -- node[above,midway] {\textbf{sat}} (dec);
        \draw[->] (smt.south) -- node[left,midway] {\textbf{unsat}} (cond);
        \draw[->] (dec.east) -- node[right,pos=1] {$M_p$} ++(0.5cm,0); 
    \end{tikzpicture}
    \caption{Graphical description of Discrete Task Planning. }
    \label{fig:diag2}
\end{figure}

This process is illustrated in Figure \ref{fig:diag2}. We first generate a discrete abstraction $M_d$ of the continuous system $M_c$. We then encode the abstract system $M_d$, specification $\gls*{rtlformula}$, and counterexamples $M_{cex}$ into a Boolean formula which, is verified with an SAT/SMT solver. If this formula is satisfiable, we decode the satisfying evaluation of the variables to a Kripke structure $M_p$, which represents behaviors of $M_{d}$ that satisfies the specification. Otherwise, we increase the problem bound $K$ until we determine that no solution exists. 

%%%%%%%%%%%%%%%%%%%%%%%%%%%%%%%%%%%%%%%%%%%%%%%%%%%%%%%%%%%%%%%%%%%%%%%%%%%%%%%%%%%%%%%%%%%%%%%%%%%%%%%%%%%%%%%%%%%%%%%%%%%%%%%%%%%%%%%%%%%%%%%%%%%%%%%%%%%%%%%%%%%%%%%%%%
\subsection{Abstraction}
%%%%%%%%%%%%%%%%%%%%%%%%%%%%%%%%%%%%%%%%%%%%%%%%%%%%%%%%%%%%%%%%%%%%%%%%%%%%%%%%%%%%%%%%%%%%%%%%%%%%%%%%%%%%%%%%%%%%%%%%%%%%%%%%%%%%%%%%%%%%%%%%%%%%%%%%%%%%%%%%%%%%%%%%%%

The discrete abstraction $M_d$ is a Kripke structure $M_d = \langle S_d, Act_d, T_d, I_d, L_d, \Sigma, F_d \rangle$ over a finite set of states $S_d = \{ s_1,s_2,\dots,s_{|S_d|} \}$ and an empty set of inputs $Act_d = \emptyset$. $M_d$ is formally related to the continuous system $M_c$ through the notion of a \textit{simulation relation} \cite{baier2008principles}. 

\begin{definition}\label{def:simulation}
    A relation $R \subseteq S \times S^\prime$ is a \textit{simulation relation} from Kripke structure $M$ to $M^\prime$ (i.e., $M^\prime$ simulates $M$) if:
    \begin{enumerate}
        \item for every $s_0 \in I$ there exists $s_0^\prime \in I^\prime$ such that $(s_0,s_0^\prime) \in R$;
        \item for every $s \in F$ there exists $s' \in F_d$ such that $(s,s') \in R$;
        \item for every $(s,s^\prime) \in R$ we have $L(s) = L^\prime(s^\prime)$;
        \item for every $(s_k,s_k^\prime) \in R$ we have that:
        \item[] for every $a_k \in Act$ with $(s_k,a_k,s_{k+1}) \in T$ there exists $a_k^\prime \in Act^\prime$ with $(s_k^\prime,a_k^\prime,s_{k+1}^\prime) \in T^\prime$ satisfying $(s_{k+1},s_{k+1}^\prime) \in R$.
    \end{enumerate}     
\end{definition} 

We will construct the discrete abstraction $M_d$ such that $M_d$ simulates $M_c$. Intuitively, this means that the discrete model $M_d$ can express every behavior of $M_c$ with respect to the specification. 

To construct the discrete abstraction, note that an \gls*{rtl} formula \gls*{rtlformula} can be used to construct convex polytopic partitions $\gls*{poly}$ such that the same predicates hold for all continuous states $\gls*{statec}_k \in \gls*{poly}$. An example of such partitions is shown in Figure \ref{fig:ex1}. We use this property to construct a discrete abstraction $M_d$ which simulates $M_c$, as follows:

\begin{enumerate}
    \item Construct a finite set of polytopes $\gls*{polyseq}$ representing the state formulas $\gls*{stateformula} \in cl(\gls*{rtlformula})$, where each polytope $\gls*{poly}^{\boldsymbol{p}} \in \gls*{polyseq}$ represent a state space such that $L_c(\gls*{statec}) = \boldsymbol{p}$ for all $\gls*{statec} \in \gls*{poly}^{\boldsymbol{p}}$, and $\boldsymbol{p}$ is a set of symbols. Algorithm \ref{alg:abs} describes how to construct these polytopes. 
    \item Each polytope in $\gls*{polyseq}$  corresponds to a state $s \in S_d$. We denote the operation that recovers the polytope of a state $s \in S_d$ by $\gls*{poly}(s) = \gls*{poly}^{\boldsymbol{p}}$. The initial state $I_d = \{ s_0 \}$ is the state $s_0 \in S_d$ such that $\gls*{statecini} \in \gls*{poly}(s_0)$. Moreover, the accepting states $F_d = S_d$. The labeling function $L_d : S_d \mapsto \Sigma$ is defined such that $p \in L_d(s)$ if and only if $L_c(\gls*{statec}) = \boldsymbol{p}$ for all $\gls*{statec} \in \gls*{poly}(s)$. Observe that $L_d(s) = \boldsymbol{p}$ of $\gls*{poly}(s) = \gls*{poly}^{\boldsymbol{p}}$.   
    \item Finally, $(\gls*{poly}_k,\gls*{poly}_{k+1}) \in T_d$ if and only if $\gls*{poly}_k$ and $\gls*{poly}_{k+1}$ are adjacent. We call two polytopes adjacent if they are equal or if their intersection is a polytope of dimension $\gls*{statecnb}-1$. For example, a polytope with dimension $\gls*{statecnb}-1$ is a line if $\gls*{statecnb} = 2$ or a plane if $\gls*{statecnb} = 3$.
\end{enumerate}

\begin{algorithm}
    \caption{Partition from State Formulas}\label{alg:abs}
    \begin{algorithmic}
    \REQUIRE {\small $M_{c}$}
    \STATE {\small $\gls*{polyseq} \gets \gls*{poly}^{\emptyset} = \gls*{statecdom}$;} 
    \FOR {\small$\gls*{stateformula}_i \in cl(\gls*{rtlformula})$;}
        \FOR {\small$\gls*{poly}^{p_i} \in toPolytopes(\gls*{stateformula}_i)$;}
            \STATE {\small $\gls*{polyseq}_{p_i} \gets \gls*{poly}^{p_i}$;$\gls*{polyseq}^\prime \gets \emptyset$;$\gls*{polyseq}_{p_i}^\prime \gets \emptyset$;} 
            \FOR {\small$\gls*{poly}^{\boldsymbol{p}} \in \gls*{polyseq}$;}
                \FOR{\small$\gls*{poly}_j^{p_i} \in \gls*{polyseq}_{p_i}$;}
                    \STATE {\small $\gls*{polyseq}^\prime \gets \gls*{polyseq}^\prime \cup \gls*{poly}^{\boldsymbol{p}\cup\{p_i\}} = \gls*{poly}^{\boldsymbol{p}} \cap \gls*{poly}^{p_i}$;} 
                    \STATE {\small $\gls*{polyseq}^\prime \gets \gls*{polyseq}^\prime \cup \gls*{poly}_1^{\boldsymbol{p}} \cup \dots \cup \gls*{poly}_{N_1}^{\boldsymbol{p}} = \gls*{poly}^{\boldsymbol{p}} \setminus \gls*{poly}^{p_i}$;} 
                    \STATE {\small $\gls*{polyseq}_{p_i}^\prime \gets \gls*{polyseq}^\prime \cup \gls*{poly}_1^{p_i} \cup \dots \cup \gls*{poly}_{N_2}^{p_i} = \gls*{poly}_j^{p_i} \setminus \gls*{poly}^{\boldsymbol{p}}$;} 
                \ENDFOR
                \STATE {\small $\gls*{polyseq}_{p_i} \gets \gls*{polyseq}_{p_i}^\prime$;} 
            \ENDFOR
            \STATE {\small $\gls*{polyseq} \gets \gls*{polyseq}^\prime$;} 
        \ENDFOR
    \ENDFOR
    \STATE {\small \RETURN $\gls*{polyseq}$;} 
    \end{algorithmic}
\end{algorithm}

This procedure always generates a transition system $M_d$ which simulates $M_c$.
\begin{proposition}\label{prop:simulation}
    Given an \gls*{rtl} formula \gls*{rtlformula} and a transition system $M_c$, there exists a transition system $M_d$ and a simulation relation $R_d$ such that any run $\gls*{run}_c$ of $M_c$ which satisfies the simulation relation $R_d$ for a run $\gls*{run}_d$ that satisfies the formula \gls*{rtlformula} also satisfies the formula \gls*{rtlformula}, i.e., $\gls*{run}_c \gls*{sat} \gls*{rtlformula}$ only if $\gls*{run}_d \gls*{sat} \gls*{rtlformula}$ and $(\gls*{statec}_k,s_k) \in R_d$ for all $k \in \mathbb{N}$.
\end{proposition}
\begin{proof}
    We prove the existence of $M_d$ and $R_d$ by construction. Using the proposed abstraction, the initial state contains the initial state, i.e., $\gls*{statecini} \in \gls*{poly}(s_0)$. Hence, the condition 1 of Definition \ref{def:simulation} is satisfied. We also define $F_d = S_d$ such that condition 2 of Definition \ref{def:simulation} is satisfied. By construction, the labeling function $L_d : S_d \mapsto 2^{\Sigma}$ satisfies the condition 3 of Definition \ref{def:simulation}. Under Assumption \ref{a:adjacent}, the relation $T_d$ satisfies condition 4 of Definition \ref{def:simulation}. Finally, by condition 3 of Definition \ref{def:simulation}, $\gls*{run}_c \gls*{sat} \gls*{rtlformula}$ only if $\gls*{run}_d \gls*{sat} \gls*{rtlformula}$ and $(\gls*{statec}_k,s_k) \in R_d$ for all $k \in \mathbb{N}$.
\end{proof}

\begin{example}\label{ex:running} %{\color{red} \bf Too brief, need more details...}
As a motivating example, consider a double integrator in $\mathbb{R}^2$ with a sampling time of $1s$ \added{(i.e. $\ddot{x} = u$ where $x$ and $v = \dot{x}$ are state variables $\boldsymbol{x} = [x,v]^\intercal$ and $u$ is the input variable)}. The system starts at $\gls*{statecini} = [1,-5.5]^\intercal$ and the input is bounded, i.e., \added{$|u| \leq 2$}. This problem is inspired by \cite[Example 11.5]{belta2017formal}. The system must avoid a forbidden region in state space, visit one of two regions of interest, and reach a target, as illustrated in Figure \ref{fig:ex1}. 

We define $18$ atomic propositions which specify predicates representing unsafe states $a$ (blue region), a target $b$ (red region), and areas of interest $c$ (yellow regions). The specification can be written as $\gls*{rtlformula} = \gls*{always} \left( (\idstluntil{\gls*{neg} a}{}{b}) \gls*{and} (\idstluntil{\neg b}{}{c}) \right)$. The first part of this formula ($\neg a \boldsymbol{U} b$) ensures that for all times before reaching the target $b$, the unsafe state $a$ is avoided. The second part of the formula ($\neg b \boldsymbol{U} c$) specifies that region $c$ must be visited before region $b$. We choose this example to illustrate our approach because it considers an underactuated system with an unbounded \gls*{rtl} formula. To the best of our knowledge, no existing trajectory synthesis algorithm from \gls*{rtl} specifications can solve this problem with provable soundness and completeness. 

\begin{figure*}[htp]
    \centering
    \subfloat[Workspace]{\label{fig:ex1}
        \includegraphics[width=0.4\textwidth]{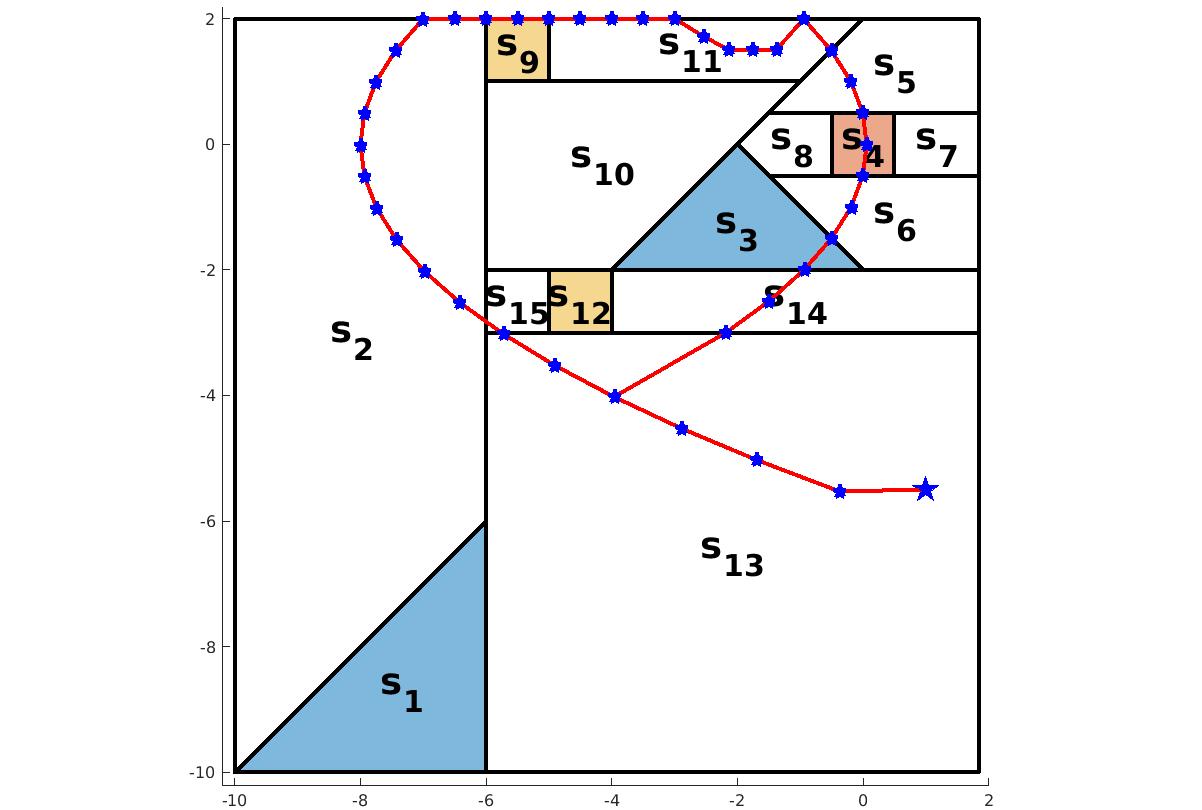}}
        \subfloat[Kripke Structure]{\label{fig:ex1_abs_kripke}
    \scalebox{0.9}{
    \begin{tikzpicture}[->,>=stealth',shorten >=1pt,auto,node distance=0.5cm,
                    semithick]
  \tikzstyle{every state}=[inner sep=0pt, minimum size=15pt]

  \node[state,initial,accepting] (s13) {$s_{13}$};
  \draw[->] (s13) to[out=80,in=100,loop] (s13); 
  
  \node[state,accepting,below right=1cm and 0.25cm of s13,label={below:a}] (s1) {$s_1$};
  \draw[->] (s1) to[out=80,in=100,loop] (s1); 
  \draw[->] (s1) to[out=-170,in=-100] (s13);
  \draw[->] (s13) to[out=-80,in=170] (s1); 
  
  \node[state,accepting,above right=0.5cm and 0.25cm of s1] (s2) {$s_2$};
  \draw[->] (s2) to[out=80,in=100,loop] (s2); 
  \draw[->] (s2) to[out=170,in=-60] (s13);
  \draw[->] (s13) to[out=-80,in=-170] (s2); 
  \draw[->] (s2) to[out=-150,in=60] (s1);
  \draw[->] (s1) to[out=80,in=-170] (s2); 

  \node[state,accepting,below right=1cm and 0.25cm of s2,label={below:b}] (s9) {$s_9$};
  \draw[->] (s9) to[out=80,in=100,loop] (s9); 
  \draw[->] (s9) to[out=-170,in=-100] (s2);
  \draw[->] (s2) to[out=-80,in=170] (s9); 
  
  \node[state,accepting,above right=1cm and 0.25cm of s9] (s10) {$s_{10}$};
  \draw[->] (s10) to[out=80,in=100,loop] (s10); 
  \draw[->] (s10) to[out=-150,in=60] (s9);
  \draw[->] (s9) to[out=80,in=-170] (s10); 
  \draw[->] (s10) to[out=170,in=10] (s2);
  \draw[->] (s2) to[out=-10,in=-170] (s10); 
  
  \node[state,accepting,right=0.25cm of s9] (s11) {$s_{11}$};
  \draw[->] (s11) to[out=80,in=100,loop] (s11); 
  \draw[->] (s11) to[out=170,in=10] (s9);
  \draw[->] (s9) to[out=-10,in=-170] (s11); 
  \draw[->] (s11) to[out=70,in=-80] (s10);
  \draw[->] (s10) to[out=-100,in=110] (s11); 

  \node[state,accepting,right=0.25cm of s10,label={below:a}] (s3) {$s_3$};
  \draw[->] (s3) to[out=80,in=100,loop] (s3); 
  \draw[->] (s3) to[out=170,in=10] (s10);
  \draw[->] (s10) to[out=-10,in=-170] (s3); 

  \node[state,accepting,right=0.25cm of s11] (s5) {$s_5$};
  \draw[->] (s5) to[out=80,in=100,loop] (s5); 
  \draw[->] (s5) to[out=170,in=10] (s11);
  \draw[->] (s11) to[out=-10,in=-170] (s5); 
  \draw[->] (s5) to[out=130,in=-70] (s10);
  \draw[->] (s10) to[out=-60,in=110] (s5); 

  \node[state,accepting,below right=0.25cm and 0.5cm of s5] (s7) {$s_7$};
  \draw[->] (s7) to[out=80,in=100,loop] (s7); 
  \draw[->] (s7) to[out=180,in=-70] (s5);
  \draw[->] (s5) to[out=-50,in=160] (s7); 

  \node[state,accepting,above right=1.5cm and 0.25cm of s3,label={below:b}] (s12) {$s_{12}$};
  \draw[->] (s12) to[out=80,in=100,loop] (s12); 
  \draw[->] (s12) to[out=150,in=70] (s13);
  \draw[->] (s13) to[out=60,in=160] (s12); 
  \draw[->] (s12) to[out=-150,in=60] (s10);
  \draw[->] (s10) to[out=70,in=-160] (s12); 
  \draw[->] (s12) to[out=-130,in=60] (s3);
  \draw[->] (s3) to[out=70,in=-140] (s12); 
  
  \node[state,accepting,above right=0.5cm and 0.25cm of s2] (s15) {$s_{15}$};
  \draw[->] (s15) to[out=80,in=100,loop] (s15); 
  \draw[->] (s15) to[out=150,in=30] (s13);
  \draw[->] (s13) to[out=20,in=160] (s15); 
  \draw[->] (s15) to[out=-160,in=70] (s2);
  \draw[->] (s2) to[out=50,in=-140] (s15); 
  \draw[->] (s15) to[out=-10,in=110] (s10);
  \draw[->] (s10) to[out=130,in=-30] (s15); 
  \draw[->] (s15) to[out=40,in=180] (s12);
  \draw[->] (s12) to[out=170,in=50] (s15); 
  
  \node[state,accepting,below right=0.5cm and 0.25cm of s12] (s14) {$s_{14}$};
  \draw[->] (s14) to[out=80,in=100,loop] (s14); 
  \draw[->] (s14) to[out=150,in=50] (s13);
  \draw[->] (s13) to[out=40,in=160] (s14); 
  \draw[->] (s14) to[out=-160,in=40] (s3);
  \draw[->] (s3) to[out=50,in=-170] (s14); 
  \draw[->] (s14) to[out=110,in=-40] (s12);
  \draw[->] (s12) to[out=-50,in=120] (s14); 

  \node[state,accepting,below right=0.5cm and 0.25cm of s14] (s6) {$s_6$};
  \draw[->] (s6) to[out=80,in=100,loop] (s6); 
  \draw[->] (s6) to[out=120,in=-50] (s14); 
  \draw[->] (s14) to[out=-40,in=110] (s6);
  \draw[->] (s6) to[out=170,in=10] (s3); 
  \draw[->] (s3) to[out=-10,in=-170] (s6);
  \draw[->] (s6) to[out=-140,in=60] (s7); 
  \draw[->] (s7) to[out=50,in=-130] (s6);

  \node[state,accepting,below=0.5cm of s6,label={below:c}] (s4) {$s_4$};
  \draw[->] (s4) to[out=80,in=100,loop] (s4); 
  \draw[->] (s4) to[out=70,in=-70] (s6); 
  \draw[->] (s6) to[out=-120,in=120] (s4);
  \draw[->] (s4) to[out=-150,in=30] (s7); 
  \draw[->] (s7) to[out=10,in=-130] (s4);

  \node[state,accepting,above right=0.25cm and 0.25cm of s5] (s8) {$s_8$};
  \draw[->] (s8) to[out=80,in=100,loop] (s8); 
  \draw[->] (s8) to[out=170,in=-40] (s10); 
  \draw[->] (s10) to[out=-50,in=-170] (s8);
  \draw[->] (s8) to[out=140,in=-40] (s3); 
  \draw[->] (s3) to[out=-50,in=150] (s8);
  \draw[->] (s8) to[out=50,in=-150] (s6); 
  \draw[->] (s6) to[out=-160,in=60] (s8);
  \draw[->] (s8) to[out=10,in=170] (s4); 
  \draw[->] (s4) to[out=-170,in=-10] (s8);

\end{tikzpicture}}
}
    \caption{Graphical representation of the illustrative example. (a) The workspace after the abstraction and a solution for $\gls*{statecini} = [1,-5.5]^\intercal$. 
    (b) A graphical representation of the Kripke structure $M_d$ that abstracts $M_c$.}
    \label{fig:ex1_problem}
\end{figure*}
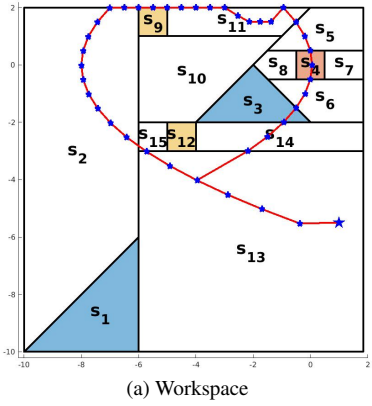
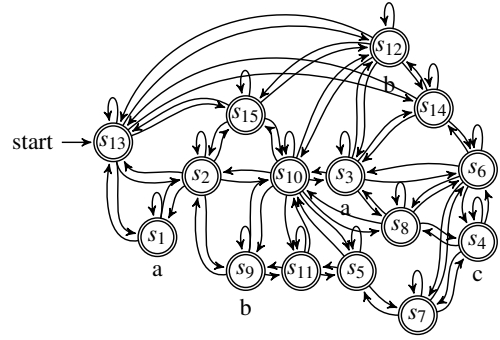

Algorithm \ref{alg:abs} starts with the workspace and the specification, and generates the discrete abstraction $M_d$. The associated polytopic partitions are shown in Figure \ref{fig:ex1} while the Kripke structure $M_d$ is illustrated in Figure \ref{fig:ex1_abs_kripke}.
\end{example}

%%%%%%%%%%%%%%%%%%%%%%%%%%%%%%%%%%%%%%%%%%%%%%%%%%%%%%%%%%%%%%%%%%%%%%%%%%%%%%%%%%%%%%%%%%%%%%%%%%%%%%%%%%%%%%%%%%%%%%%%%%%%%%%%%%%%%%%%%%%%%%%%%%%%%%%%%%%%%%%%%%%%%%%%%%
\subsection{RTL equivalent Kripke Structure}
%%%%%%%%%%%%%%%%%%%%%%%%%%%%%%%%%%%%%%%%%%%%%%%%%%%%%%%%%%%%%%%%%%%%%%%%%%%%%%%%%%%%%%%%%%%%%%%%%%%%%%%%%%%%%%%%%%%%%%%%%%%%%%%%%%%%%%%%%%%%%%%%%%%%%%%%%%%%%%%%%%%%%%%%%%

Instead of passing a single satisfying run to the continuous planning, we construct a set of \gls*{rtl} language equivalent runs in the form of a Kripke structure. In the continuous planning, this structure essentially defines a sequence of convex constraints, which, if satisfied, guarantees the  specification $\gls*{rtlformula}$.

\begin{definition}
An \gls*{rtl} equivalent Kripke structure $M^\prime$ from a run $\gls*{run}_d$ is a Kriple structure where every run of this structure satisfies the same \gls*{rtl} formulas that the run $\gls*{run}_d$ satisfies. This means: $\gls*{run}_d \gls*{sat} \gls*{rtlformula}$ if and only if $\gls*{run}^\prime \gls*{sat} \gls*{rtlformula}$ for all runs $\gls*{run}^\prime$ of $M^\prime$.
\end{definition}

We illustrate the process that constructs an \gls*{rtl} equivalent Kripke structure $M^\prime$ from a run $\gls*{run}_d$  in Algorithm \ref{alg:dec}. 
If the loop exists, $M^\prime$ is a fair Kripke structure, meaning that it generates infinite runs with a loop.
\added{Intuitively, a dynamical system may take more time to pass through the polytopic constraints of a discrete plan. Thus, we construct a Kripke structure that represent these longer runs but preserving the \gls*{rtl} equivalence.}%\added{Intuitively, this algorithm constructs a Kripke structure of all runs that repeat or introduce a previous state before a state of the run, denoted as repeat and backward operations.}
\begin{algorithm}
    \caption{Construct an \gls*{rtl} equivalent Kripke structure $M^\prime$.}\label{alg:dec}
    \begin{algorithmic}[1]
    \REQUIRE {\small $\gls*{run}_d = s_0\dots s_{L-1}(s_L\dots s_K)^\omega, M_d$}
    \ENSURE{\small $M^\prime = \langle S^\prime, T^\prime, I^\prime, L^\prime, \Sigma, F^\prime \rangle$;} 
    \STATE {\small $i \gets 0$; $S^\prime \gets s_i^\prime$; $T^\prime \gets \{(s_i^\prime,s_i^\prime)\}$; $I_p \gets s_i^\prime$; $L^\prime(s_i^\prime) \gets s_0$;}
    \FOR {\small$k = 1$ \TO $K-1$}
        \STATE {\small $i=i+1$; $S^\prime \gets s_i^\prime$; $L^\prime(s_i^\prime) \gets s_k$;}
        \STATE {\small $T^\prime \gets  \{(s_i^\prime,s_{i-1}^\prime),(s_{i-1}^\prime,s_i^\prime),(s_i^\prime,s_i^\prime)\}$;}
        \IF {\small $k = L$}
        \STATE {\footnotesize $i=i+1$; $S^\prime \gets s_i^\prime$; $L^\prime(s_i^\prime) \gets s_k$; $T^\prime \gets  \{(s_{i-2}^\prime,s_i^\prime),(s_{i-1}^\prime,s_i^\prime),(s_i^\prime,s_i^\prime)\}$;}
        \ENDIF
    \ENDFOR
    \STATE {\small $i=i+1$; $S^\prime \gets s_i^\prime$; $L^\prime(s_i^\prime) \gets s_K$;$T^\prime \gets  \{(s_{i-1}^\prime,s_i^\prime)\}$;}
    \IF {\small $L < K$}
        \STATE {\small $i=i+1$; $S^\prime \gets s_i^\prime$; $L^\prime(s_i^\prime) \gets s_k$; }
        \STATE {\small $T^\prime \gets  \{(s_i^\prime,s_{i-1}^\prime),(s_{i-1}^\prime,s_i^\prime),(s_{i}^\prime,s_{i+1}^\prime),(s_{i+1}^\prime,s_{i+1}^\prime),(s_{L+1}^\prime,s_{i+1}^\prime)$}
        \STATE {\small \hspace{5.9cm} $,(s_{i+1}^\prime,s_{L+1}^\prime)\}$;}
    \ENDIF
    \STATE {\small $F^\prime \gets s_i^\prime$;}
    \end{algorithmic}
\end{algorithm}

\begin{example}\label{ex:dplan}
    Consider the system of Example \ref{ex:running} again. A satisfying run is $\gls*{run}_d = s_{13}(s_2s_9s_{11}s_5s_4s_6s_{14}s_{13})^\omega$.
    Fig.~\ref{fig:ex1-plan} shows a graphical representation of $M^\prime$ for this run. First, note that the labelling function $L^\prime : S^\prime \rightarrow S_d$ maps each state to a state from the abstraction $M_d$. At each step of Algorithm \ref{alg:dec}, we generate a new state $s_i^\prime$ that has a self-loop and back and forward transitions. When there is a loop, we introduce proxy states ($s_2^\prime$ and $s_{10}^\prime$ in Fig.~\ref{fig:ex1-plan}), to ensure that runs of $M^\prime$ follow the same loop as the original run. This forces the system to pass through states $s_3^\prime$, $s_4^\prime$, $\dots$, $s_8^\prime$ in order to visit the accepting state $s_9^\prime$ infinitely often.
\end{example}

\begin{figure}
    \centering
    \subfloat[line 1]{
    \scalebox{0.85}{
    \begin{tikzpicture}[->,>=stealth',shorten >=1pt,auto,node distance=0.75cm,
                    semithick]
  \tikzstyle{every state}=[inner sep=0pt, minimum size=0pt]

  \node[state,label={below:$s_{13}$},initial]      (s0)                    {$s_0^\prime$};

  \draw[->] (s0) to[out=45,in=90,loop] (s0);  
\end{tikzpicture}}}     
    \subfloat[$k = 1$, lines 3 and 4]{
    \scalebox{0.85}{
    \begin{tikzpicture}[->,>=stealth',shorten >=1pt,auto,node distance=0.75cm,
                    semithick]
  \tikzstyle{every state}=[inner sep=0pt, minimum size=0pt]

  \node[state,label={below:$s_{13}$},initial]      (s0)                    {$s_0^\prime$};
  \node[state,label={below:$s_2$},right of= s0]    (s1)                    {$s_1^\prime$};

  \draw[->] (s0) to[out=45,in=90,loop] (s0);  
  
  \draw[->] (s0) to[out=45,in=135]     (s1);
  \draw[->] (s1) to[out=-135,in=-45]   (s0);
  \draw[->] (s1) to[out=45,in=90,loop] (s1);  
\end{tikzpicture}}}  
    \subfloat[$k = 1$, lines 6]{
    \scalebox{0.85}{
    \begin{tikzpicture}[->,>=stealth',shorten >=1pt,auto,node distance=0.75cm,
                    semithick]
  \tikzstyle{every state}=[inner sep=0pt, minimum size=0pt]

  \node[state,label={below:$s_{13}$},initial]      (s0)                    {$s_0^\prime$};
  \node[state,label={below:$s_2$},right of= s0]    (s1)                    {$s_1^\prime$};
  \node[state,label={below:$s_2$},right of= s1]    (s2)                    {$s_2^\prime$};

  \draw[->] (s0) to[out=45,in=90,loop] (s0);  
  
  \draw[->] (s0) to[out=45,in=135]     (s1);
  \draw[->] (s1) to[out=-135,in=-45]   (s0);
  \draw[->] (s1) to[out=45,in=90,loop] (s1);  
  
  \draw[->] (s0) to[out=45,in=90,looseness=2]     (s2);
  \draw[->] (s1) --     (s2);
  \draw[->] (s2) to[out=45,in=90,loop] (s2);  
\end{tikzpicture}}}  

    \subfloat[$k=7$, lines 3 and 4]{
    \scalebox{0.85}{
    \begin{tikzpicture}[->,>=stealth',shorten >=1pt,auto,node distance=0.75cm,
                    semithick]
  \tikzstyle{every state}=[inner sep=0pt, minimum size=0pt]

  \node[state,label={below:$s_{13}$},initial]      (s0)                    {$s_0^\prime$};
  \node[state,label={below:$s_2$},right of= s0]    (s1)                    {$s_1^\prime$};
  \node[state,label={below:$s_2$},right of= s1]    (s2)                    {$s_2^\prime$};
  \node[state,label={below:$s_9$},right of= s2]    (s3)                    {$s_3^\prime$};
  \node[state,label={below:$s_{11}$},right of= s3] (s4)                    {$s_4^\prime$};
  \node[state,label={below:$s_5$},right of= s4]    (s5)                    {$s_5^\prime$};
  \node[state,label={below:$s_4$},right of= s5]    (s6)                    {$s_6^\prime$};
  \node[state,label={below:$s_6$},right of= s6]    (s7)                    {$s_7^\prime$};
  \node[state,label={below:$s_{14}$},right of= s7] (s8)                    {$s_8^\prime$};

  \draw[->] (s0) to[out=45,in=90,loop] (s0);  
  \draw[->] (s0) to[out=45,in=135]     (s1);
  \draw[->] (s0) to[out=45,in=90,looseness=2]     (s2);
  \draw[->] (s1) to[out=-135,in=-45]   (s0);
  
  \draw[->] (s1) to[out=45,in=90,loop] (s1);  
  \draw[->] (s1) --     (s2);

  \draw[->] (s2) to[out=45,in=90,loop] (s2);  
  \draw[->] (s2) to[out=45,in=135]     (s3);
  \draw[->] (s3) to[out=-135,in=-45]   (s2);
  \draw[->] (s3) to[out=45,in=90,loop] (s3);  
  \draw[->] (s3) to[out=45,in=135]     (s4);
  \draw[->] (s4) to[out=-135,in=-45]   (s3);
  \draw[->] (s4) to[out=45,in=90,loop] (s4);  
  \draw[->] (s4) to[out=45,in=135]     (s5);
  \draw[->] (s5) to[out=-135,in=-45]   (s4);
  \draw[->] (s5) to[out=45,in=90,loop] (s5);  
  \draw[->] (s5) to[out=45,in=135]     (s6);
  \draw[->] (s6) to[out=-135,in=-45]   (s5);
  \draw[->] (s6) to[out=45,in=90,loop] (s6);  
  \draw[->] (s6) to[out=45,in=135]     (s7);
  \draw[->] (s7) to[out=-135,in=-45]   (s6);
  \draw[->] (s7) to[out=45,in=90,loop] (s7);  
  \draw[->] (s7) to[out=45,in=135]     (s8);
  \draw[->] (s8) to[out=-135,in=-45]   (s7);
  \draw[->] (s8) to[out=45,in=90,loop] (s8);

\end{tikzpicture}}}

    \subfloat[lines 7-12]{
    \scalebox{0.85}{
    \begin{tikzpicture}[->,>=stealth',shorten >=1pt,auto,node distance=0.75cm,
                    semithick]
  \tikzstyle{every state}=[inner sep=0pt, minimum size=0pt]

  \node[state,label={below:$s_{13}$},initial]      (s0)                    {$s_0^\prime$};
  \node[state,label={below:$s_2$},right of= s0]    (s1)                    {$s_1^\prime$};
  \node[state,label={below:$s_2$},right of= s1]    (s2)                    {$s_2^\prime$};
  \node[state,label={below:$s_9$},right of= s2]    (s3)                    {$s_3^\prime$};
  \node[state,label={below:$s_{11}$},right of= s3] (s4)                    {$s_4^\prime$};
  \node[state,label={below:$s_5$},right of= s4]    (s5)                    {$s_5^\prime$};
  \node[state,label={below:$s_4$},right of= s5]    (s6)                    {$s_6^\prime$};
  \node[state,label={below:$s_6$},right of= s6]    (s7)                    {$s_7^\prime$};
  \node[state,label={below:$s_{14}$},right of= s7] (s8)                    {$s_8^\prime$};
  \node[state,label={below:$s_{13}$},right of= s8,accepting] (s9)                    {$s_9^\prime$};
  \node[state,label={below:$s_{13}$},right of= s9] (s10)                    {$\scriptstyle s_{10}^\prime$};

  \draw[->] (s0) to[out=45,in=90,loop] (s0);  
  \draw[->] (s0) to[out=45,in=135]     (s1);
  \draw[->] (s0) to[out=45,in=90,looseness=2]     (s2);
  \draw[->] (s1) to[out=-135,in=-45]   (s0);
  
  \draw[->] (s1) to[out=45,in=90,loop] (s1);  
  \draw[->] (s1) --     (s2);

  \draw[->] (s2) to[out=45,in=90,loop] (s2);  
  \draw[->] (s2) to[out=45,in=135]     (s3);
  \draw[->] (s3) to[out=-135,in=-45]   (s2);
  \draw[->] (s3) to[out=45,in=90,loop] (s3);  
  \draw[->] (s3) to[out=45,in=135]     (s4);
  \draw[->] (s4) to[out=-135,in=-45]   (s3);
  \draw[->] (s4) to[out=45,in=90,loop] (s4);  
  \draw[->] (s4) to[out=45,in=135]     (s5);
  \draw[->] (s5) to[out=-135,in=-45]   (s4);
  \draw[->] (s5) to[out=45,in=90,loop] (s5);  
  \draw[->] (s5) to[out=45,in=135]     (s6);
  \draw[->] (s6) to[out=-135,in=-45]   (s5);
  \draw[->] (s6) to[out=45,in=90,loop] (s6);  
  \draw[->] (s6) to[out=45,in=135]     (s7);
  \draw[->] (s7) to[out=-135,in=-45]   (s6);
  \draw[->] (s7) to[out=45,in=90,loop] (s7);  
  \draw[->] (s7) to[out=45,in=135]     (s8);
  \draw[->] (s8) to[out=-135,in=-45]   (s7);
  \draw[->] (s8) to[out=45,in=90,loop] (s8);  
  \draw[->] (s8) to[out=45,in=135]     (s9);
  \draw[->] (s9) to[out=-135,in=-45]   (s8);

  \draw[->] (s9) to[out=90,in=90,looseness=0.5]   (s2);
  \draw[->] (s9) to[out=45,in=135]     (s10);
  \draw[->] (s10) to[out=45,in=90,loop] (s10);  
  \draw[->] (s10) to[out=-135,in=-45,looseness=0.75]     (s2);
  \draw[->] (s2) to[out=-60,in=-120,looseness=0.75]     (s10);

\end{tikzpicture}}}
    \caption{Graphical representation of the construction of an \gls*{rtl} equivalent Kripke structure $M^\prime$ from the run $\gls*{run}_d = s_{13}(s_2s_9s_{11}s_5s_4s_6s_{14}s_{13})^\omega$.}
    \label{fig:ex1-plan}
\end{figure}
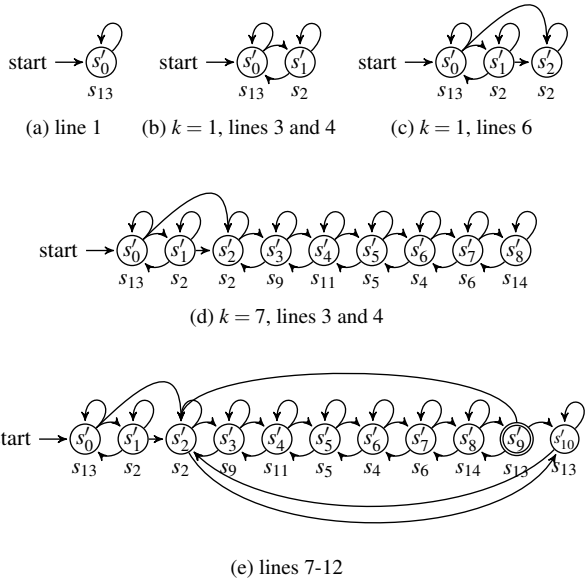

We can prove that Algorithm \ref{alg:dec} is sound and complete:

\begin{proposition}\label{prop:dec}
    Algorithm \ref{alg:dec} constructs a Kripke structure $M^\prime$ if and only if there exists an \gls*{rtl} Kripke structure $M^\prime$ from the discrete abstraction run $\gls*{run}_d$.
\end{proposition}
\begin{proof}
We will prove by structural induction. We first define the path constructors that generate paths of $M^\prime$ recursively: $repeat\_at(i,\gls*{run}^{\prime}_{parent})$ and $backward\_at(i,\gls*{run}^{\prime}_{parent})$. Given a run $\gls*{run}^{\prime}_{parent} = (s_0^\prime)^{\ell_0^\prime}\dots(s_{K^\prime}^\prime)^{\ell_{K^\prime}^\prime}$ of $M^\prime$ with length $\sum_{i=0}^{K^\prime} \ell_i^\prime$, these operators allow us to construct another run of $\gls*{run}^{\prime}_{child}= (s_0^\prime)^{\ell_0^{\prime\prime}}\dots(s_{K^{\prime\prime}}^\prime)^{\ell_{K^{\prime\prime}}^{\prime\prime}}$ of $M^\prime$ with length $\sum_{i=0}^{K^{\prime\prime}} \ell_i^{\prime\prime}$ which is longer, i.e., $\sum_{i=0}^{K^{\prime\prime}} \ell_i^{\prime\prime} > \sum_{i=0}^{K^\prime} \ell_i^\prime$, as follows:
\begin{enumerate}
    \item $repeat\_at(i,\gls*{run}^{\prime}_{parent}) = \gls*{run}^{\prime}_{child} \in (S^\prime)^* : \gls*{run}^{\prime}_{child} = (s_0^\prime)^{\ell_0^\prime}\dots(s_i^\prime)^{\ell_i^\prime+1}$ $\dots (s_{K^\prime}^\prime)^{\ell_{K^\prime}^\prime}$;
    \item $backward\_at(i,\gls*{run}^{\prime}_{parent}) = \gls*{run}^{\prime}_{child} \in (S^\prime)^* : \gls*{run}^{\prime}_{child} = (s_0^\prime)^{\ell_0^\prime}\dots(s_i^\prime)^{\ell_i^\prime}$ $(s_{i-1}^\prime)^{1}(s_i^\prime)^{1}\dots (s_{K^\prime}^\prime)^{\ell_{K^\prime}^\prime}$ if $(s_i^\prime,s_{i-1}^\prime) \in T^\prime$
\end{enumerate}

These operators generate all possible runs of $M^\prime$ because they represent the transitions in $T^\prime$ generated by Algorithm \ref{alg:dec}. In particular, the operator $repeat\_at(i,\gls*{run}^{\prime}_{parent})$ represents the transitions $(s_i^\prime,s_i^\prime) \in T^\prime$, and the operator $backward\_at(i,^\prime\gls*{run}^{\prime}_{parent})$ the transitions $(s_i^\prime,s_j^\prime) \in T^\prime$ when $i \neq j$. In other words, the combination of these operators over the shortest run of $M^\prime$ generates all possible runs of this structure. 

Now, we can start the structural induction proof. First, note that the shortest run $\gls*{run}^{\prime^\star}$ of $M^\prime$ has path equal to the path of the satisfying $\gls*{run}_d$ of $M_d$. As a result, $\gls*{run}^{\prime^\star} \vDash \gls*{rtlformula}$.

Next, we assume that the run $\gls*{run}^{\prime}_{parent}$ satisfies the specification, (i.e., $\gls*{run}^\prime_{parent} \vDash \gls*{rtlformula}$). Any run generated by the operators $repeat\_at(i,\gls*{run}^{\prime}_{parent})$ \added{and $backward\_at(i,^\prime\gls*{run}^{\prime}_{parent})$} satisfies the specification because \added{the \gls*{rtl} semantics permits repetitions}. \added{For example, consider that  $\gls*{run}^{\prime}_{parent} \vDash_k \gls*{rtlformula}_1 \gls*{until} \gls*{rtlformula}_2$. Thus, there exists an instant $k^\prime > k$ such that $\gls*{run}^{\prime}_{parent} \vDash_{k^\prime} \gls*{rtlformula}_2$ and $\gls*{run}^{\prime}_{parent} \vDash_{k^\prime - 1} \gls*{rtlformula}_1$. As a result, if we apply the backward operator at instant $k^\prime$, this means that that this formula is satisfied at instants $k \leq k^{\prime\prime} \leq k^\prime + 2$ because $\gls*{run}^{\prime}_{parent} \vDash_{k^\prime + 1} \gls*{rtlformula}_1$ and $\gls*{run}^{\prime}_{parent} \vDash_{k^\prime + 2} \gls*{rtlformula}_2$.}  This holds for the release operator as well. Therefore, the proposition holds by structural induction.
\end{proof}

\begin{remark}
\added{Intuitively, this is analogous to oscillation behaviors exhibited by continuous dynamical systems. The discrete plan indicates regions that the continuous trajectory should evolve through. Sometimes we may need to revisit a region to drive the system trajectory to a goal region, which requires the backward operation. }
\end{remark}

%%%%%%%%%%%%%%%%%%%%%%%%%%%%%%%%%%%%%%%%%%%%%%%%%%%%%%%%%%%%%%%%%%%%%%%%%%%%%%%%%%%%%%%%%%%%%%%%%%%%%%%%%%%%%%%%%%%%%%%%%%%%%%%%%%%%%%%%%%%%%%%%%%%%%%%%%%%%%%%%%%%%%%%%%%
\subsection{Counterexamples}\label{sec:counterexamples}
%%%%%%%%%%%%%%%%%%%%%%%%%%%%%%%%%%%%%%%%%%%%%%%%%%%%%%%%%%%%%%%%%%%%%%%%%%%%%%%%%%%%%%%%%%%%%%%%%%%%%%%%%%%%%%%%%%%%%%%%%%%%%%%%%%%%%%%%%%%%%%%%%%%%%%%%%%%%%%%%%%%%%%%%%%

\added{The discrete abstraction simulates the system; thus, this abstraction has runs that do not render valid runs of the system, which we denote as dynamically infeasible runs. So, we also identify an \gls*{iis} \cite{chinneck1991locating} for Problem \ref{prob:1}. An \gls*{iis} defines an infeasible subset of constraints such that removing any one constraint renders the subset feasible. We call the constraints in this \gls*{iis} counter-examples.}
\begin{definition}
    Given a feasibility problem with a set of constraints $\mathcal{C}$, an \emph{Irreducibly Inconsistent Set} $\mathcal{I}$ is a subset $\mathcal{I} \subseteq \mathcal{S}$ such that: (1) the feasibility problem with the constraint set $\mathcal{I}$ is infeasible; and (2) $\forall c \in \mathcal{I}$, the feasibility problem with constraint set $\mathcal{I}\setminus \{c\}$ is feasible. 
\end{definition}

\added{Similarly to the discrete plan, we represent these counter-examples as an RTL equivalent Kripke Structure. When we identify an abstraction run that is not feasible, we construct a Kripke structure representing its RTL equivalent runs. In summary, this structure runs are all runs that we can generate using the repeat and backward operators from Proposition \ref{prop:dec}.  
}

\added{We denote the set of counter-examples found so far as $M_{cex}$, where each counter-example $M_{cex}^i \in M_{cex}$ is a Kripke structure. We construct this Kripke structure using Algorithm \ref{alg:dec} in the same way that we construct discrete plans. Therefore, we can discard all unfeasible runs by the product of the discrete abstraction $M_d$ and the complement of the counter-examples $M_{cex}^i$, i.e., $M_d \times M_{cex}^{ic}$.}

\begin{example}\label{ex:bsc:cex}
Consider the system of Example \ref{ex:running} but starting at $\gls*{statecini} = [-4,-8]^\intercal$ instead. Since the input is bounded, we will not be able to generate a run $\gls*{run}_c$ from $s_{13}$ to $s_2$ without passing through $s_1$. Consequently, any run of $M_d$ with prefix $s_{13} (s_{13})^*s_2$, shown in Fig.~\ref{fig:ex1feas}, is dynamically infeasible. We pass the shortest run of the prefix $s_{13} (s_{13})^*s_2$ (i.e., $s_{13}s_2$) to Algorithm \ref{alg:dec}, Fig.~\ref{fig:ex1feas_line1} and \ref{fig:ex1feas_line7}. Then, we add a suffix to this Kripke structure to accept all prefixes, i.e., $s_{13} (s_{13})^*s_2(s_1+\dots+s_{15})^*$, Fig.~\ref{fig:ex1feas_pref}.
\begin{figure}[htp]
    \centering
    % \subfloat[Workspace]{\centering
    %     \includegraphics[width=0.24\textwidth]{running_unfeas.jpg}}
    \subfloat[line 1]{\label{fig:ex1feas_line1}
    \centering
    \scalebox{0.85}{
    \begin{tikzpicture}[->,>=stealth',shorten >=1pt,auto,node distance=0.75cm,
                    semithick]
  \tikzstyle{every state}=[inner sep=0pt, minimum size=0pt]
  \node[state,label={below:$s_{13}$},initial]      (s0)                    {$s_0^\prime$};

  \draw[->] (s0) to[out=45,in=90,loop] (s0);  

  \end{tikzpicture}}}
    \subfloat[line 7]{\label{fig:ex1feas_line7}
    \centering
    \scalebox{0.85}{
    \begin{tikzpicture}[->,>=stealth',shorten >=1pt,auto,node distance=0.75cm,
                    semithick]
  \tikzstyle{every state}=[inner sep=0pt, minimum size=0pt]
  \node[state,label={below:$s_{13}$},initial]      (s0)                    {$s_0^\prime$};
  \node[state,label={below:$s_2$},right of= s0,accepting]    (s1)                    {$s_1^\prime$};

  \draw[->] (s0) to[out=45,in=90,loop] (s0);  
  \draw[->] (s0) to[out=45,in=135]     (s1);

  \end{tikzpicture}}}
  \subfloat[All prefixes]{\label{fig:ex1feas_pref}
    \centering
    \scalebox{0.85}{
    \begin{tikzpicture}[->,>=stealth',shorten >=1pt,auto,node distance=0.75cm,
                    semithick]
  \tikzstyle{every state}=[inner sep=0pt, minimum size=0pt]
  \node[state,label={below:$s_{13}$},initial]      (s0)                    {$s_0^\prime$};
  \node[state,label={below:$s_2$},right of= s0,accepting]    (s1)                    {$s_1^\prime$};
  \node[state,label={below:$S_d$},right of= s1,accepting]    (s2)                    {$s_2^\prime$};

  \draw[->] (s0) to[out=45,in=90,loop] (s0);  
  \draw[->] (s0) to[out=45,in=135]     (s1);
  \draw[->] (s1) to[out=45,in=135]     (s2);
  \draw[->] (s2) to[out=45,in=90,loop] (s2);  

\end{tikzpicture}}}
    \caption{Example of a counter-example construction for the discrete plan $s_{13}(s_2s_9s_{11}s_5s_4s_6s_{14}s_{13})^\omega$ when starting at $\gls*{statecini} = [-4,-8]^\intercal$.} 
    \label{fig:ex1feas}
\end{figure}
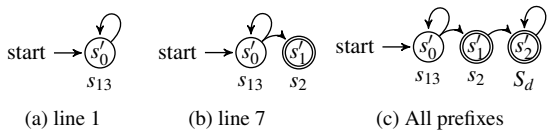

\end{example}
%{\color{red} \bf Mke it a formal example ...}

%%%%%%%%%%%%%%%%%%%%%%%%%%%%%%%%%%%%%%%%%%%%%%%%%%%%%%%%%%%%%%%%%%%%%%%%%%%%%%%%%%%%%%%%%%%%%%%%%%%%%%%%%%%%%%%%%%%%%%%%%%%%%%%%%%%%%%%%%%%%%%%%%%%%%%%%%%%%%%%%%%%%%%%%%%
\subsection{Encoding}
%%%%%%%%%%%%%%%%%%%%%%%%%%%%%%%%%%%%%%%%%%%%%%%%%%%%%%%%%%%%%%%%%%%%%%%%%%%%%%%%%%%%%%%%%%%%%%%%%%%%%%%%%%%%%%%%%%%%%%%%%%%%%%%%%%%%%%%%%%%%%%%%%%%%%%%%%%%%%%%%%%%%%%%%%%

%{\color{red} \bf Suggest to add this as a "Remark" at the end of this subsection... }

Given the abstract system $M_d$, the specification $\gls*{rtlformula}$, and a Kripke structure $M_{cex}$ representing counterexamples, we encode the problem of finding a satisfying run of $M_d$ as a Boolean satisfiability problem \added{using techniques presented in \cite{heljanko2005incremental,schuppan2006linear}}. To do this, we separate the encoding into three components: the abstract system encoding $ \encform{M_d}_K$, the specification encoding $\encform{\gls*{rtlformula}}_K$, and the counter-example encoding $\encform{M_{cex}}_K$. These encodings can be combined into one Boolean formula
\begin{equation}\label{eq:encoding}
    \encform{M_d,\gls*{rtlformula},M_{cex},K} := \encform{M_d}_K \gls*{and} \encform{\gls*{rtlformula}}_K \gls*{and} \encform{M_{cex}}_K.
\end{equation}

Particularly, these encodings are indexed by the bound $K$. The basic idea is to search for short (small $K$) solutions first, then incrementally increase this bound until we reach a value, after which the specification $\gls*{rtlformula}$ is unsatisfiable. This iterative structure allows us to harness incremental solvers like Z3 \cite{de2008z3} to efficiently find satisfying evaluations for specifications with a high $K$. 

The encodings for the abstract system, specification, and counterexamples are presented below. 

\begin{remark}
The linear nature of this encoding means that the number of constraints increases linearly with $K$. Additionally, its incremental nature allows incremental SAT/SMT solvers to use learned clauses from previous iterations, drastically improving performance.
\end{remark}

%%%%%%%%%%%%%%%%%%%%%%%%%%%%%%%%%%%%%%%%%%%%%%%%%%%%%%%%%%%%%%%%%%%%%%%%%%%%%%%%%%%%%%%%%%%%%%%%%%%%%%%%%%%%%%%%%%%%%%%%%%%%%%%%%%%%%%%%%%%%%%%%%%%%%%%%%%%%%%%%%%%%%%%%%%
\subsubsection{Abstract System}\label{subsubsec:abstract_sys}
%%%%%%%%%%%%%%%%%%%%%%%%%%%%%%%%%%%%%%%%%%%%%%%%%%%%%%%%%%%%%%%%%%%%%%%%%%%%%%%%%%%%%%%%%%%%%%%%%%%%%%%%%%%%%%%%%%%%%%%%%%%%%%%%%%%%%%%%%%%%%%%%%%%%%%%%%%%%%%%%%%%%%%%%%%
%To encode $M_d$ as a Boolean formula, we need to capture the infinite behavior of $M_d$ with a finite representation. We do this in one of two ways: either a finite run $\gls*{run}_{d}^K = s_0\dots s_K$ is sufficient to represent all its infinite extensions, or a finite run in $K$-loop form $\gls*{run}_d^{K,L} = s_0\dots s_{L-1}(s_{L}\dots s_{K})^\omega$ captures the infinite behavior.

%With this in mind, the abstract system $M_d$ can be represented symbolically as a Boolean formula that captures finite paths with length $K$,
The abstract system $M_d$ can be represented symbolically as a Boolean formula that captures finite paths with length $K$,
\begin{equation*}
\encform{M_d}_K := I(\hat{s}_0) \gls*{and} \bigwedge_{i=0}^{K-1} T(\hat{s}_i,\hat{s}_{i+1}) \gls*{and} F(\hat{s}_K),
\end{equation*}  
where $\hat{s}_k$ models states $s_k$ as bit vectors, $I(\hat{s}_0)$ and $F(\hat{s}_K)$ are Boolean formulas ensuring that the state $\hat{s}_0$ is the one of the initial states and $\hat{s}_K$ is one of accepting states, and $T(\hat{s}_i,\hat{s}_{i+1})$ encodes the requirements of a transition from $\hat{s}_i$ to $\hat{s}_{i+1}$. We also define the transformation $\chi(\hat{s}_k) = s_k$ which returns the state $s_k$ in $M_d$ corresponding to an evaluation of $\hat{s}_k$.

\begin{example}
Considering the system and specification of Example \ref{ex:running}. The states can be abstracted with a vector of four bits, i.e., $\hat{s}_k \in \{ 0, 1, \dots, 15 \}$ and $\chi(\hat{s}_k) = s_{\hat{s}_k-1}$. The initial conditions are encoded by the formula $I(\hat{s}_0) = \hat{s}_0 = 12$. Since, in this example, we do not restrict the final state, the formula that represents the accepting states is trivially true. Finally, the transitions at instant $k$ are encoded by a formula that is the disjunction of sub-formulas representing each valid transition $(s_i,s_j)$ as follows: $\hat{s}_k = \chi^{-1}(s_i) \gls*{and} \hat{s}_{k+1} = \chi^{-1}(s_j)$.
\end{example}

\subsubsection{Specification} %{\color{red} \bf Check the notations, not very consistent ...}
Following the RTL semantics, we encode the specification $\gls*{rtlformula}$ recursively by considering \textit{formula variables} $\encform{\gls*{rtlformula}^\prime}_k$. For every subformula $\gls*{rtlformula}^\prime \in cl(\gls*{rtlformula})$, a variable $\encform{\gls*{rtlformula}^\prime}_k$ is introduced and interpreted as true if and only if $\gls*{run}_d \gls*{sat}_k \gls*{rtlformula}^\prime$. 

The Boolean encoding of propositional operators in \gls*{rtl} formulas for the instants $0 \leq k \leq K$ is as follows:
\begin{itemize}
\item {\small $\encform{p}_k \gls*{iff} \bigvee\limits_{s_d \in \{ s_d \in S_d : p \in L_d(s_d) \} }  \hat{s}_k = \chi(s_d)$} ,
\item {\small $\encform{\gls*{rtlformula}_1 \gls*{and} \gls*{rtlformula}_2}_k \gls*{iff} \encform{\gls*{rtlformula}_1}_k \gls*{and} \encform{\gls*{rtlformula}_2}_k$},
\item {\small $\encform{\gls*{rtlformula}_1 \gls*{or} \gls*{rtlformula}_2}_k \gls*{iff} \encform{\gls*{rtlformula}_1}_k \gls*{or} \encform{\gls*{rtlformula}_2}_k$}.
\end{itemize} 

For subformulas with temporal operators, we refer to future formula variables to ensure the temporal behavior. Thus, we have,
\begin{itemize}
\item \added{{\small $\encform{\idstluntil{\gls*{rtlformula}_1}{}{\gls*{rtlformula}_2}}_k \gls*{iff} \Big(\encform{\gls*{rtlformula}_2}_k \gls*{or} \big(\encform{\gls*{rtlformula}_1}_k \gls*{and} \encform{\idstluntil{\gls*{rtlformula}_1}{}{\gls*{rtlformula}_2}}_{k+1}\big)\Big)$}},
\item \added{{\small $\encform{\idstlrelease{\gls*{rtlformula}_1}{}{\gls*{rtlformula}_2}}_k \gls*{iff} \Big(\encform{\gls*{rtlformula}_2}_k \gls*{and} \big( \encform{\gls*{rtlformula}_1}_k) \gls*{or} \encform{\idstlrelease{\gls*{rtlformula}_1}{}{\gls*{rtlformula}_2}}_{k+1}\big)\Big)$}}.
\end{itemize}

\begin{example}
Consider the formula $x \leq 1 \gls*{until} x \geq 1$. The encoding for a length $K$ is,
\begin{align*}
    & \encform{x \leq 1 \gls*{until} x \geq 1}_0 = \\
    & \bigwedge_{k=0}^K \left(
    \begin{aligned}
    & (\encform{x \leq 1}_k \gls*{iff} p_{1,k}) \gls*{and} (\encform{x \geq 1}_k \gls*{iff} p_{2,k}) \gls*{and} \\
    & \Bigg(\encform{x \leq 1 \gls*{until} x \geq 1}_k \gls*{iff}  \Big(\encform{x \geq 1}_k  \\
    & \hspace{2cm} \gls*{or} \big(\encform{x \leq 1}_k \gls*{and} \encform{x \leq 1 \gls*{until} x \geq 1}_{k+1}\big)\Big)\Bigg)
    \end{aligned}
    \right),
\end{align*}
where $p_{i,k}$ are Boolean variables encoding $p_i \in L_d(\chi(\bar{s}_k))$.
\end{example}

We still need to take into account the possible infinite behavior encoded in the specification. As mentioned above, we can model infinite behavior as a finite run with a loop. For this reason, we introduce Boolean variables $l_k$, which are true only if the loop starts at instant $k$, $In_k$, which holds only if the instant $k$ is within a loop, and $Exists$, which holds if and only if a loop exists. \added{Furthermore, we divide the loop-related constraints as the base, $ K $-independent (i.e., $1 \leq k \leq K$), and $ K $-dependent constraints.  We assert the base constraint only once, at the initialization. At each bound increment, we delete the old $ K $-dependent constraint assertions and assert $ K $-independent and $ K $-dependent constraints. This procedure allows us to keep most of the constraints between steps and harness incremental solver techniques.}

Consequently, the following constraint $\encform{Loop}_K$ defines a loop:
\begin{itemize}
\item \textbf{Base}: {\small $l_0 \gls*{iff} \gls*{false}$, and $In_0 \gls*{iff} \gls*{false}$},
\item $\boldsymbol{1 \leq k \leq K}$: {\small $l_k \gls*{implies} \hat{s}_{k-1}=\hat{s}_E$, $In_k \gls*{iff} In_{k-1} \gls*{or} l_k$, and $In_{k-1} \gls*{implies} \gls*{neg}l_k$},
\item $\boldsymbol{K}$\textbf{-dependent}: {\small $Exists \gls*{iff} In_K$, and $\hat{s}_E \gls*{iff} \hat{s}_K$}.
\end{itemize} 
Note that we introduced a proxy variable $\hat{s}_E$ to separate the $K$-dependent constraints. 

Additionally, to compensate for change in the bound $K$, we define a set of constraints for the last state $\encform{LastState}_K$. For each subformula $\gls*{rtlformula}^\prime \in cl(\gls*{rtlformula})$, we have,
\begin{itemize}
\item \textbf{Base}: {\small $\gls*{neg}Exists \gls*{implies} \big( \encform{\gls*{rtlformula}^\prime}_L \gls*{iff} \gls*{false} \big)$},
\item $\boldsymbol{1 \leq k \leq K}$: {\small $l_k \gls*{implies} \big( \encform{\gls*{rtlformula}^\prime}_L \gls*{iff} \encform{\gls*{rtlformula}^\prime}_k \big)$},
\item $\boldsymbol{K}$\textbf{-dependent}: {\small $\encform{\gls*{rtlformula}^\prime}_E \gls*{iff} \encform{\gls*{rtlformula}^\prime}_K$ and $\encform{\gls*{rtlformula}^\prime}_L \gls*{iff} \encform{\gls*{rtlformula}^\prime}_{K+1}$}.
\end{itemize}

The encoding above allows the case where $\encform{\idstluntil{\gls*{rtlformula}_1}{}{\gls*{rtlformula}_2}}$ is true for all indices of the loop even if $\encform{\gls*{rtlformula}_2}$ is not at any index of the loop, which violates the \gls*{rtl} semantics. As a result, we introduce the eventually constraints $\encform{EventRTL}_K$ and its auxiliary formula variables $\encform{\gls*{eventually} \gls*{rtlformula}_2}_E$ and $\encform{\idstluntil{\gls*{rtlformula}_1}{}{\gls*{rtlformula}_2}}_E$ such that,
\begin{itemize}
\item \textbf{Base}: {\small  $Exists \gls*{implies} \big( \encform{\idstluntil{\gls*{rtlformula}_1}{}{\gls*{rtlformula}_2}}_E \gls*{implies} \encform{\gls*{eventually} \gls*{rtlformula}_2}_E \big)$, and $\encform{\gls*{eventually} \gls*{rtlformula}_2}_0 \gls*{iff} \gls*{false}$},
\item $\boldsymbol{1 \leq k \leq K}$: {\small $\encform{\gls*{eventually} \gls*{rtlformula}_2}_k \gls*{iff} \encform{\gls*{eventually} \gls*{rtlformula}_2}_{k-1}$, or $\big( In_k \gls*{and} \encform{\gls*{rtlformula}_2}_k \big)$},
\item $\boldsymbol{K}$\textbf{-dependent}: {\small $\encform{\gls*{eventually} \gls*{rtlformula}_2}_E \gls*{iff} \encform{\gls*{eventually} \gls*{rtlformula}_2}_K$}.
\end{itemize} 

Similarly, for subformulas $\idstlrelease{\gls*{rtlformula}_1}{}{\gls*{rtlformula}_2} \in cl{\gls*{rtlformula}}$, we have,
\begin{itemize}
\item \textbf{Base}:  {\small $Exists \gls*{implies} \big( \encform{\idstlrelease{\gls*{rtlformula}_1}{}{\gls*{rtlformula}_2}}_E \gls*{if} \encform{\gls*{always} \gls*{rtlformula}_2}_E \big)$, and $\encform{\gls*{always} \gls*{rtlformula}_2}_0 \gls*{iff} \gls*{true}$},
\item $\boldsymbol{1 \leq k \leq K}$: {\small $\encform{\gls*{always} \gls*{rtlformula}_2}_k \gls*{iff} \encform{\gls*{always} \gls*{rtlformula}_2}_{k-1}$, and $\big( \gls*{neg} In_k \gls*{or} \encform{\gls*{rtlformula}_2}_k \big)$},
\item $\boldsymbol{K}$\textbf{-dependent}: {\small $\encform{\gls*{always} \gls*{rtlformula}_2}_E \gls*{iff} \encform{\gls*{always} \gls*{rtlformula}_2}_K$}.
\end{itemize} 

Putting these pieces together, the resulting Boolean formula that encodes an \gls*{rtl} formula is
\begin{align*}
\encform{\gls*{rtlformula}}_K := & \encform{Loop}_K \gls*{and} \encform{LastState}_K \gls*{and} \\
& \hspace{2cm}\encform{EventRTL}_K \gls*{and} \encform{\gls*{rtlformula}}_0.
\end{align*}

\begin{example}
Consider the specification from Example \ref{ex:running}. First, observe that the temporal operator always $\gls*{always} \varphi_1$  is equivalent to a release formula $\perp \gls*{release} \varphi_1$, where $\varphi_1 = \left( (\idstluntil{\gls*{neg} a}{}{b}) \gls*{and} (\idstluntil{\neg b}{}{c}) \right)$. From the encoding  $\encform{\perp \gls*{release}  \varphi}_k$, it follows that $\encform{\perp \gls*{release}  \varphi}_{K+1}$ should always hold true. Then, because of the encoding $\encform{LastState}_K$, there must always exists a loop. An example of a satisfying run is $s_{13}(s_{12}s_{14}s_{6}s_4s_6s_{14})^\omega$. Notice that there is a loop that will enforce that the states $s_{12}$ and $s_4$ ($c$ and $b$) will always be visited infinitely often. Note that invalid runs have one of the following characteristics: (i) they do not have a loop, (ii) they have a loop but do not visit $s_4$ and one of the states $s_9$ and $s_{12}$ inside the loop, (iii) or they visit $s_1$ or $s_3$. 
\end{example}

\subsubsection{Counter-Examples}

\added{We can encode these counterexamples in the same way as we encoded the Kripke structure $M_d$ of the abstract system:  $\encform{M_{cex}}_K = \bigwedge_{i = 1}^{|M_{cex}|} \gls*{neg} \encform{ M_{cex}^i}_K$, where $\encform{ M_{cex}^i}_K = I_{cex}(\hat{s}_0) \gls*{and} \bigwedge_{i=0}^{K-1} T_{cex}(\hat{s}_i,\hat{s}_{i+1}) \gls*{and} F_{cex}(\hat{s}_K)$.}

\subsection{Soundness}
We establish the correctness of the proposed encoding by the following proposition:
\begin{proposition}\label{prop:dplan_soundness}
    Given a Kripke structure $M_d$, an \gls*{rtl} formula \gls*{rtlformula}, and a set of counterexamples $M_{cex}$, a run $\gls*{run}$ of $M_d$ satisfies the specification (i.e., $\gls*{run}_d \gls*{sat} \gls*{rtlformula}$) if there exists $K \in \mathbb{N}$ such that the encoding $\encform{M_d,\gls*{rtlformula},M_{cex},K}$ is satisfiable.
\end{proposition}
\begin{proof}
     We first prove that if there exists a bound $K$ such that  $\encform{M_d,\gls*{rtlformula},M_{cex},K}$ is satisfiable, then $\gls*{run}_d \gls*{sat}_K \gls*{rtlformula}$. Then, the proposition follows because if the run $\gls*{run}_d$ satisfies the specification with bound $K$, then it satisfies the specification, i.e., $\gls*{run}_d \gls*{sat}_K \gls*{rtlformula}$ implies that $\gls*{run}_d \gls*{sat} \gls*{rtlformula}$. The sufficiency of checking runs with only $K$ steps follows from the loop structure described in Section\ref{subsubsec:abstract_sys}.

    It is easily seen that the constraint $\encform{M_d}_K \wedge \encform{M_{cex}}_K$ encodes all valid finite runs of model $M_d \setminus M_{cex}$ of length $K$. Moreover, the loop constraints $\encform{Loop}_K$ ensure two cases of satisfying runs: (a) when a loop exists, there will exists an unique index $1 \leq j \leq K$ such that $\hat{s}_j = \hat{s}_K$ determinines when the loop starts, and (b) when there is no loop, the run $\gls*{run}_d$ is a prefix of $M_d$. 
    
    Now, we prove that for any subformula $\gls*{rtlformula}^\prime \in cl(\gls*{rtlformula})$, $\gls*{run}_d \gls*{sat}_k \gls*{rtlformula}^\prime$ if $\encform{\gls*{rtlformula}^\prime}_k$ is true. It is easy to see that the claim holds for the cases where \gls*{rtlformula} is an atomic proposition. Moreover, the claim also trivially holds by induction when \gls*{rtlformula} is a boolean function of atomic propositions. We still need to prove the claim for formulas with temporal operators. In fact, the encoding follows the one-step identity of temporal operators $\gls*{until}$ and $\gls*{release}$. Specifically, $\encform{\idstluntil{\gls*{rtlformula}_1}{}{\gls*{rtlformula}_2}}_k$ is true if either (i) $\exists k^\prime \in [k..K]$ s.t. $\encform{\gls*{rtlformula}_2}_i$ is true and $\encform{\gls*{rtlformula}_1}_{k^{\prime\prime}}$ is true for all $k \leq k^{\prime\prime} \leq k^\prime$, or (ii) the proxy variable $\encform{\idstluntil{\gls*{rtlformula}_1}{}{\gls*{rtlformula}_2}}_{K+1}$ is true and $\encform{\gls*{rtlformula}_1}_{k^{\prime\prime}}$ is true for all $k \leq k^{\prime\prime} \leq K$. When no loop exists, $\encform{\gls*{rtlformula}}_{K+1}$ is false; consequently, the claim holds by induction, when the claim holds for the subformulas. Now, when the loop exists \added{starting at index $1 \leq j \leq K$}, the constraint $\encform{LastState}_K$ ensures that the proxy variables $\encform{\gls*{rtlformula}^\prime}_{K+1} \equiv \encform{\gls*{rtlformula}^\prime}_{j}$. Moreover, the constraint $\encform{EventRTL}_K$ ensures that there exists $j \leq k \leq K$ s.t. $\encform{\idstluntil{\gls*{rtlformula}_1}{}{\gls*{rtlformula}_2}}_k$ is true only if there exists $j \leq k^\prime \leq K$ s.t. $\encform{\gls*{rtlformula}_2}_{k^\prime}$ is true. \added{Reciprocally, the enconding of the temporal operator $\encform{\idstluntil{\gls*{rtlformula}_1}{}{\gls*{rtlformula}_2}}_k$ ensures that $\gls*{rtlformula}_1$ holds until $k^\prime$, i.e., $\encform{\gls*{rtlformula}_1}_k$ holds true for $0 \leq k < k^\prime$.} Thus, the claim holds by induction. The same reasoning applied to the case of the temporal operator $\gls*{release}$. Consequently, any subformula $\gls*{rtlformula}^\prime \in cl(\gls*{rtlformula})$, $\gls*{run}_d \gls*{sat}_k \gls*{rtlformula}^\prime$ if $\encform{\gls*{rtlformula}^\prime}_k$ is true. Therefore, there exists a run $\gls*{run}$ of $M_d \setminus M_{cex}$ such that  $\gls*{run}_d \gls*{sat} \gls*{rtlformula}$ if there exists $K$ such that $\encform{M_d,\gls*{rtlformula},M_{cex},K}$ holds.
\end{proof}

%%%%%%%%%%%%%%%%%%%%%%%%%%%%%%%%%%%%%%%%%%%%%%%%%%%%%%%%%%%%%%%%%%%%%%%%%%%%%%%%%%%%%%%%%%%%%%%%%%%%%%%%%%%%%%%%%%%%%%%%%%%%%%%%%%%%%%%%%%%%%%%%%%%%%%%%%%%%%%%%%%%%%%%%%%
\subsection{Completeness}
%%%%%%%%%%%%%%%%%%%%%%%%%%%%%%%%%%%%%%%%%%%%%%%%%%%%%%%%%%%%%%%%%%%%%%%%%%%%%%%%%%%%%%%%%%%%%%%%%%%%%%%%%%%%%%%%%%%%%%%%%%%%%%%%%%%%%%%%%%%%%%%%%%%%%%%%%%%%%%%%%%%%%%%%%%
The proposed incremental encoding allows us to use incremental SAT/SMT solvers and determine when to stop increasing the bound $ K $. In this regard, our procedure for completeness is based on an inductive procedure proposed in \cite{schuppan2006linear}. The main idea is to check if a longer discrete run that satisfies the specification may still exist by removing the $K$-dependent constraints from the encoding. The longest initialized loop-free run\added{, i.e., a run where the initial state of the run is an initial state of the system, and all states are distinct,} that satisfies the specification is called the \textit{recursive diameter} and is used as the upper bound for the completeness threshold. We use a straightforward encoding of this loop-free run predicate, whose size is quadratic with the bound (i.e., $O(K^2)$) \cite{schuppan2006linear}. 

First, we define a \textit{completeness formula} $\encform{M_d,\gls*{rtlformula},M_{cex},K}_C$ which consists of exactly the encoding $\encform{M_d,\gls*{rtlformula},M_{cex},K}$ with all $K$-dependent constraints removed. Intuitively, $\encform{M_d,\gls*{rtlformula},M_{cex},K}_C$ is satisfied only if there exists runs $\gls*{run}_d$ of $M_d$ that satisfy the specification with length $K$ or longer bounds. Moreover, we propose a \textit{simple run} formula which is satisfiable for only initialised loop-free runs. Let $\encform{\hat{s}_{\gls*{rtlformula}}}_k$ be a bit vector of values of all formula variables $\encform{\gls*{rtlformula}}_k$, the simple run predicate is defined as follows:
\begin{equation*}
    \encform{SR}_K := \bigwedge_{0 \leq i < j \leq K} \big( \hat{s}_i \neq \hat{s}_j \gls*{or} \gls*{neg} In_i \gls*{or} \gls*{neg} In_j \gls*{or} \encform{\hat{s}_{\gls*{rtlformula}}}_i \neq \encform{\hat{s}_{\gls*{rtlformula}}}_j \big).
\end{equation*}

Now, we prove the completeness of this encoding. Note that as an intermediate result, we determine some $ K $ above which increasing it does not change the satisfaction.

\begin{proposition}\label{prop:dplan_completeness}
    Given a Kripke structure $M_d$, an \gls*{rtl} formula \gls*{rtlformula}, and a set of counterexamples $M_{cex}$, there is no run $\gls*{run}_d$ in $M_d$ discarding counterexamples that satisfies the specification if for some $K\geq 0$ $\encform{M_d,\gls*{rtlformula},M_{cex},K}_C \gls*{and} \encform{SR}_K$ is unsatisfiable and either $K = 0$ or $\encform{M_d,\gls*{rtlformula},M_{cex},K-1}$ is unsatisfiable.  
\end{proposition}
\begin{proof}
    First, note that new counterexamples will not change the satisfiability of past checking iterations. Moreover, as mentioned above, $\encform{M_d,\gls*{rtlformula},M_{cex},K^\prime}_C$ is unsatisfiable implies that $\encform{M_d,\gls*{rtlformula},M_{cex},K}_C$ is unsatisfiable for all $K \geq K^\prime$.
    
    Consider that for some $K\geq 0$ $\encform{M_d,\gls*{rtlformula},M_{cex},K}_C$ $\gls*{and} \encform{SR}_K$ is unsatisfiable and either $\encform{M_d,\gls*{rtlformula},M_{cex},K-1}$ is unsatisfiable or $K = 0$. If $K=0$, it implies that $\encform{M_d,\gls*{rtlformula},M_{cex},0}_C$ is unsatisfiable because $\encform{SR}_0$ is empty; thus, $\encform{M_d,\gls*{rtlformula},M_{cex},K}_C$ is unsatisfiable for all $k \geq 0$. Now, notice that if there exist $K$ such that $\encform{M_d,\gls*{rtlformula},M_{cex}^\prime,K}$ is satisfiable, then $\encform{M_d,\gls*{rtlformula},M_{cex}^\prime,k}_C\gls*{and} \encform{SR}_K$ is satisfiable for $k \leq K$. Thus, if  $\encform{M_d,\gls*{rtlformula},M_{cex}^\prime,K}\gls*{and} \encform{SR}_K$ is unsatifiable, $\encform{M_d,\gls*{rtlformula},M_{cex}^\prime,k}$ is unsatisfiable for all $k \geq K$. Using Proposition \ref{prop:dplan_soundness}, we conclude that there is no run $\gls*{run}_d$ in $M_d$ without counterexamples such that satisfies the specification.
\end{proof}

%{\color{red}  Strange statement ... For further details on the correctness of this encoding, we refer the interested reader to \cite{schuppan2006linear}, which provides more thorough proofs of soundness and completeness for a similar encoding.  }

%%%%%%%%%%%%%%%%%%%%%%%%%%%%%%%%%%%%%%%%%%%%%%%%%%%%%%%%%%%%%%%%%%%%%%%%%%%%%%%%%%%%%%%%%%%%%%%%%%%%%%%%%%%%%%%%%%%%%%%%%%%%%%%%%%%%%%%%%%%%%%%%%%%%%%%%%%%%%%%%%%%%%%%%%
\section{Continuous Motion Planning}\label{sec:fsearch}
%%%%%%%%%%%%%%%%%%%%%%%%%%%%%%%%%%%%%%%%%%%%%%%%%%%%%%%%%%%%%%%%%%%%%%%%%%%%%%%%%%%%%%%%%%%%%%%%%%%%%%%%%%%%%%%%%%%%%%%%%%%%%%%%%%%%%%%%%%%%%%%%%%%%%%%%%%%%%%%%%%%%%%%%%

\added{The continuous planning checks if there exists a dynamically feasible run of the system that satisfies a discrete plan. We decode a satisfying run of $M_d$ from the discrete planning and construct an \gls*{rtl} equivalent Kripke structure $M_p$ from this run.}

As described above, the discrete plan has infinite runs. As a result, we need a stop criteria to decide when to decide that the tree has no solution and generate a counter-example.
%The basic idea of continuous motion planning is to find a run of the continuous system which satisfies the specification. 
To do so, we harness the simulation relation between the continuous system $M_c$ and the discrete system $M_d$. The basic idea is to take a particular run from the discrete plan and try to find a corresponding run in $M_c$. If we cannot find such a run, we return a counter-example describing runs of $M_d$ to discard in future plans. Unlike counter-examples in traditional model checking, which prove that a system does not always satisfy the specification, this counter-example proves that there exists a prefix of the plan $M_d$ which cannot satisfy the dynamic constraints. 

Similarly to the discrete planning, the continuous planning searches for a periodic run of the form $\gls*{run}_c^{H,N} = \gls*{statec}_0\xrightarrow{\gls*{inputc}_0}\gls*{statec}_1\xrightarrow{\gls*{inputc}_1}\dots\gls*{statec}_{N-1}\xrightarrow{\gls*{inputc}_{N-1}}(\gls*{statec}_N\xrightarrow{\gls*{inputc}_N}\dots\dots \gls*{statec}_H)^\omega$. This structure, commonly used in to address infinite behavior of temporal logic specifications \cite{wolff2016optimal}, allows us to consider infinite runs with a finite representation. This requirement is necessary when computing a trajectory using \gls*{lp} solvers, and it is frequently used in trajectory synthesis approaches \cite{wolff2016optimal}. 

Now, we formally define the notion of dynamic feasibility.

\begin{definition}\label{def:feasiblerun}
    A run $\gls*{run}_c$ of the continuous system $M_c$ is dynamically feasible run of the discrete plan $M_p$, denoted as $\gls*{run}_c \in M_p$ if and only if there exists a run $\gls*{run}_p^\prime = s_{0,p}\dots s_{N-1,p}(s_{N,p}\dots s_{H,p})^\omega$ of the discrete plan $M_p$ with bound $H$ and loop starting at $N$ such that the the following problem is feasible:
    \begin{equation}\label{eq:feasiblerun}
        \begin{aligned}
        \text{find } & \gls*{run}_c \\
        \text{s.t. } & \gls*{statec}_0 = \gls*{statecini}, \gls*{statec}_H = \gls*{statec}_{N-1}, \text{ and }  \\
        \forall k \geq 0: & \gls*{statec}_{k+1} \in Post_{\gls*{tolfeas}}(\gls*{statec}_k,\gls*{inputc}_k), \text{ and } \gls*{statec}_{k+1} \in \gls*{poly}\big(L_p(s_{p,k+1})\big),
        \end{aligned}
    \end{equation}
    where:
    \begin{itemize}
        \item $Post_{\gls*{tolfeas}}(\gls*{statec}_k,\gls*{inputc}_k) := \{ \gls*{statec} \in \mathbb{R}^{\gls*{statecnb}} : \gls*{statec} = A \gls*{statec}_k + B \gls*{inputc}_k + \gls*{tolfeas}, \gls*{inputc}_k \in \gls*{inputcdom} \}$,
        %\item \added{a run $\gls*{run}_p^\prime$ of the discrete plan has the form $\gls*{run}_p^\prime = (s_0^p)^{\ell_0^\prime}\dots$ $(s_{L-1}^p)^{\ell_{L-1}^\prime} ((s_{L^\prime}^p)^{\ell_{L^\prime}^\prime} \dots \big(s_{K^\prime}^p)^{\ell_{K^\prime}^\prime}\big)^\omega$,}
        %\item \added{the run bound is $H = \sum_{i = 0}^{K^\prime} \ell_i^\prime$,}
        %\item \added{the loop starts at $N = \sum_{i = 0}^{L^\prime} \ell_i^\prime$,}
        %\item The transformation $\imath[k]$ transforms a time index $k$ into a discrete state index $i$.
        \item \added{$\gls*{statec}_{k} \in \gls*{poly}\big(L_p(s_{p,k+1})\big)$} denotes that the continuous state $\gls*{statec}_k$ is contained in the polytope corresponding to the discrete state that is the label of the discrete plan state at instant $k$, i.e., $L_p(s_{p,k+1}) \in S_d$.
    \end{itemize}
\end{definition}

The constraint \added{$\gls*{statec}_{k} \in \gls*{poly}\big(L_p(s_{p,k+1})\big)$} enforces that the continuous state resides in a corresponding polytope, which corresponds to a step of a run in the discrete plan $M_p$. We assume that inputs are bounded; thus, $\gls*{inputc}_k \in \gls*{inputcdom}$ corresponds to the input constraints. Dynamic feasibility is enforced by the constraint $\gls*{statec}_{k+1} \in Post_{\gls*{tolfeas}}(\gls*{statec}_k,\gls*{inputc}_k)$. 

\begin{remark}\label{rem:delta}
These feasibility constraints are written in terms of $\delta$-completeness \cite{gao2012delta}. This takes into consideration the fact that some systems exhibit behavior (Zeno behavior) in which improvement towards the satisfaction of given constraints is arbitrarily slow. Setting some positive but arbitrarily small value of $\delta$ allows us to find finite-length satisfying runs for such cases. Note, though, that $\delta$ can be arbitrarily small to correspond to a negligible value.
\end{remark}

\begin{example}
    These feasibility constraints are written in terms of $\delta$-completeness \cite{gao2012delta}. This completeness considers that some systems exhibit behavior (Zeno behavior) in which improvement towards the satisfaction of given constraints is arbitrarily slow. Setting some positive but arbitrarily small value of $\delta$ allows us to find finite-length satisfying runs for such cases. Note, though, that $\delta$ can be arbitrarily small to correspond to a negligible value.
\end{example}

In summary, there are two main challenges to solving the feasibility problem (\ref{eq:feasiblerun}). First, we need to determine when no feasible run $\gls*{run}_c \in M_p$ exists. Second, if no such run exists, we need to find an appropriate counterexample $M_{cex}^i$. 

\subsection{Dynamical Feasibility}

\added{We address these challenges by identifying necessary and sufficient conditions for a feasible run $\gls*{run}_c \in M_p$. The basic idea is to check the feasibility of a run by incrementally adding constraints.}

\added{We use the discrete plan run in a particular form to define the necessary and sufficient conditions. We can represent a discrete plan in the form of $\boldsymbol{\xi}_p = (s_0^p)^{\ell_0}(s_1^p)^{\ell_1}\dots (s_{L-1}^p)^{\ell_{L-1}}\big((s_L^p)^{\ell_L}\dots$ $(s_K^p)^{\ell_K}\big)^\omega$, where each two sequent states $s_i^p$ and $s_{i+1}^p$ are distinct (i.e.,  $s_i^p \neq s_{i+1}^p$). This form allows us to see the conditions as part of a polytopic tunnel in the system workspace, i.e., $\mathcal{P}(L_p(s_0^p))\mathcal{P}(L_p(s_1^p))\dots\mathcal{P}(L_p(s_K^p))$.}

\added{Furthermore, we call a segment of this run the sequence $s_{i-1}(s_i)^{\ell_i}s_{i+1}$ (or $s_i(s_i)^{\ell_i-1}s_{i+1}$ if $i = 0$), where $s_i = L_p(s_i^p)$ is the abstraction state $s_i$ that labels the plan state $s_i^p$.  }

\added{First, there exists a dynamically feasible run $\boldsymbol{\xi}_c \in M_p$ only if all prefixes of the corresponding run $\boldsymbol{\xi}_p$ are also feasible.  Formally, we denote a prefix of a run $\gls*{run}_p$ of $M_p$ as $Prefix(\gls*{run}_p,P) := \{ (s_0^p)^{\ell_0}(s_1^p)^{\ell_1}\dots (s_{P-1}^p)^{\ell_{P-1}} s_P^p \in (S_d)^* : \ell_i > 0$ for $ i = 0,\dots,P-1\}$. As a result, a prefix is feasible if (\ref{eq:feas}) is feasible for this prefix dropping the loop constraints (i.e., $\gls*{statec}_H = \gls*{statec}_{N-1}$). Therefore, this is a necessary condition for the existence of a dynamically feasible run and is defined as follows.}
\begin{definition}\label{def:feasprefix}
    A prefix $Prefix(\gls*{run}_p,P)$ is said to be \emph{feasible} if and only if the solution of the following \bf \gls*{lp} is less than $\gls*{tolfeas}$:
    \begin{equation}\label{eq:feasprefix}
        \begin{aligned}
            \max\limits_{k \in [0..H-1]} & \min\limits_{\substack{
            \gls*{inputc}_k \in \mathbb{R}^{\gls*{inputcnb}}, \\
            \gls*{statec}_k, \gls*{statec}_{k+1} \in \mathbb{R}^{\gls*{statecnb}}
            }}\| \gls*{statec}_{k+1} - A \gls*{statec}_k - B \gls*{inputc}_k \|_{\infty}\\
            \text{s.t. } &  \gls*{statec}_0 = \gls*{statecini}, \gls*{inputc}_k \in \gls*{inputcdom}, \gls*{statec}_{k+1} \in \gls*{poly}\big(L_p(s_{p,k+1})\big).
            \end{aligned}
    \end{equation}
\end{definition}

\added{Next, if a prefix of the corresponding run $\boldsymbol{\xi}_p$ is feasible, any segment $s_{i-1}(s_i^p)^{\ell_i}s_{i+1}^p$ (or $s_i(s_i^p)^{\ell_i-1}s_{i+1}^p$ if $i = 0$) in this prefix is also feasible. }
\begin{definition}\label{def:feas}
    \added{A segment $s_{i-1}(s_i)^{\ell_i}s_{i+1}$ (or $s_i(s_i)^{\ell_i-1}s_{i+1}$ if $i = 0$)} is said to be \emph{feasible} within $\ell_i$ steps
    if and only if the solution of the following \textbf{
 \gls*{lp} is less than $\gls*{tolfeas}$}:
    \begin{equation}\label{eq:feas}
        \begin{aligned}
            \max\limits_{k \in [0..\ell_i-1]} & \min\limits_{\substack{
        \gls*{inputc}_k \in \mathbb{R}^{\gls*{inputcnb}}, \\
        \gls*{statec}_k, \gls*{statec}_{k+1} \in \mathbb{R}^{\gls*{statecnb}}
        }}  \| \gls*{statec}_{k+1} - A \gls*{statec}_k - B \gls*{inputc}_k \|_{\infty} \\
        \text{s.t. } & 
        \gls*{statec}_0 \in \bar{\gls*{poly}}, \gls*{statec}_{\ell_i} \in \gls*{poly}(s_{i+1}), \\
        & \gls*{statec}_{k+1} \in \gls*{poly}(s_i) \textbf{ if } 0 < k < \ell_i, \gls*{inputc}_k \in \gls*{inputcdom},
        \end{aligned}
    \end{equation}
    where $\bar{\gls*{poly}} = \gls*{poly}(s_{i-1})$ if $i > 0$, otherwise $\bar{\gls*{poly}} = \{ \gls*{statec} = \gls*{statecini} \}$.
\end{definition}

\added{Finally, a segment $s_{i-1}(s_i)^{\ell_i}s_{i+1}$ (or $s_i(s_i)^{\ell_i-1}s_{i+1}$ if $i = 0$) is feasible only if it is reachable.  Intuitively, the reachability drop the constraint on intermediate system states to satisfies the plan state $s_i^p$ (i.e., we do not require that $\boldsymbol{x}_k \in \mathcal{P}(s_i)$ for $0 < k < \ell_i$). We formally define the reachability as follows.}
\begin{definition}\label{def:reach}
    \added{A segment $s_{i-1}(s_i)^{\ell_i}s_{i+1}$ (or $s_i(s_i)^{\ell_i-1}s_{i+1}$ if $i = 0$)} is said to be \emph{reachable} within $\ell_i$ steps if and only if the solution of the following \textbf{ \gls*{lp} is negative}:
    \begin{equation}\label{eq:reach}
        \begin{aligned}
            \max\limits_{k \in [0..\ell_i - 1]} & \min\limits_{\substack{
            \gls*{inputc}_k \in \mathbb{R}^{\gls*{inputcnb}}, \\
            \gls*{statec}_0, \gls*{statec}_{\ell_i} \in \mathbb{R}^{\gls*{statecnb}}
            }}\min   A_{\gls*{inputcdom}} \gls*{inputc}_k - \boldsymbol{b}_{\gls*{inputcdom}} \\
            \text{s.t. } & \gls*{statec}_0 \in \bar{\gls*{poly}}, \gls*{statec}_{\ell_i} \in \gls*{poly}(s_{i+1}), \\
            & \gls*{statec}_{\ell_i} - A^{\ell_i} \gls*{statec}_0 =  A^{\ell_i-1} B \gls*{inputc}_0 + A^{\ell_i - 2} B \gls*{inputc}_1 + \ldots + B \gls*{inputc}_{k^\prime},
        \end{aligned}
    \end{equation}
    where $\bar{\gls*{poly}} = \gls*{poly}(L_p(s_{i-1}^p))$ if $i > 0$, otherwise $\bar{\gls*{poly}} = \{ \gls*{statec} = \gls*{statecini} \}$.
\end{definition}

\begin{remark}
\added{Note that the constraints of (\ref{def:feasprefix}), (\ref{eq:feas}) and (\ref{eq:reach}) are linear and max-min problems can be encoded as \gls*{lp} problems using slack variables \cite{boyd2004convex}. Additionally, notice that the objective of this problem is reduce the distance of the variables to the half-spaces of a polytope. As consequence, when the solution is not positive (or greater than $\gls*{tolfeas}$), all variables in the solution are inside this polytope. Furthermore, we highlight that $\gls*{tolfeas}$ can be chosen arbitrarily small to correspond to a negligible value.}
\end{remark}

\added{We now present necessary and sufficient conditions for the existence of a feasible run $\gls*{run}_c \in M_p$:}

\begin{proposition}\label{prop:nscond}
   Given a discrete plan $M_p$ and a dynamical system $M_{c}$, there exists\added{ a feasible run $\gls*{run}_c \in M_p$} if and only if there exists \added{a discrete plan run $\gls*{run}_p$} such that:
    \begin{enumerate}
        \item \added{for $0 \leq i < K$,}
        \begin{enumerate}
            \item \added{each segment $s_{i-1}(s_i^p)^{\ell_i}s_{i+1}^p$ is reachable,} \label{enum:nscond1}
            \item \added{each segment $s_{i-1}(s_i^p)^{\ell_i}s_{i+1}^p$ is feasible,} \label{enum:nscond2}
        \end{enumerate}
        \item \added{for $0 < P \leq K$,} any prefix \added{$Prefix(\gls*{run}_p,P)$} is feasible, and \label{enum:nscond3}
        \item \added{Problem (\ref{eq:feasiblerun}) is feasible}.  \label{enum:nscond4}
    \end{enumerate}
\end{proposition}
\begin{proof}
 ($\Rightarrow$) If there exists a feasible run \added{$\gls*{run}_c \in M_p$}, then, by Definition \ref{def:feasiblerun}, there exits a run \added{$\gls*{run}_p$} such that \added{$\gls*{statec}_k \in \gls*{poly}(L_p(s_{p,k}))$} for all $k \geq 0$. Since the run $\gls*{run}_c$ satisfies the dynamical constraints of $M_c$, it proves all conditions for $\gls*{run}_p$. ($\Leftarrow$) We will prove by contradiction. Assume that there exists a run \added{$\gls*{run}_p$} such that all conditions hold but there exists no feasible run \added{$\gls*{run}_c \in M_p$}. However, all segments and prefixes of the run $\gls*{run}_p$ are feasible; thus, there must exists a $\gls*{run}_c$ that satisfies the Definition \ref{def:feasiblerun}, which contradicts the assumption and proves the theorem.
\end{proof}

\subsection{Counter-example}

The necessary and sufficient conditions for feasibility allow us to identify \added{constraints of} the \gls*{iis}\added{, i.e., the counter-examples}. \added{We use Algorithm \ref{alg:feascheck} to extract these constraints from a discrete plan run $\gls*{run}_p^\prime$.}

%We use the following facts to find an \gls*{iis} for the conditions of Proposition \ref{prop:nscond}, as illustrated in Algorithm \ref{alg:feascheck}. In summary, this algorithm considers that prefixes with bound less or equal than $P^\ast$ are \gls*{iis} if any of the proposed \gls*{lp}s return a value greater than $\lambda_{\ref{eq:reach}}^{\ast}, \lambda_{\ref{eq:feas}}^{\ast}, \lambda_{\ref{eq:feasprefix}}^{\ast}$. This algorithm is executed during the search for a node of the discrete plan $M_p$ with a dynamically feasible label. In the search, we keep track of the minimum values for the parameters $\lambda_{\ref{eq:reach}}$, $\lambda_{\ref{eq:feas}}$, and $\lambda_{\ref{eq:feasprefix}}$ for infeasible prefixes with bound $P^\ast$. Therefore, if we do not improve these parameters or do not increase the feasible prefix, we found a counter-example.

\begin{algorithm}
    \caption{Feasibility Checking}\label{alg:feascheck}
    \begin{algorithmic}[1]
    \REQUIRE {\small $\gls*{run}_p^\prime,M_{c},\delta, \lambda_{\ref{eq:reach}}^{\ast}, \lambda_{\ref{eq:feas}}^{\ast}, \lambda_{\ref{eq:feasprefix}}^{\ast}, P^{\ast} $}
    \ENSURE {\small $\lambda_{\ref{eq:reach}}, \lambda_{\ref{eq:feas}}, \lambda_{\ref{eq:feasprefix}}, P, \gls*{run}_{cex}^\prime$}
    \FOR {\small$i = 2$ \TO $K^\prime$}
        \STATE {\small $P=i$;$\lambda_{\ref{eq:reach}} := \infty$; $\lambda_{\ref{eq:feas}} := \infty$; $\lambda_{\ref{eq:feasprefix}} := \infty$; $\gls*{run}_{cex}^\prime := \emptyset$;}
        \STATE {\small $\gls*{run}_{seg} := s_0(s_0)^{\ell_0^\prime-1}s_1$; (or $\gls*{run}_{seg} := s_i(s_{i+1})^{\ell_{i+1}^\prime}s_{i+2}$; if $i = 0$)}
        \STATE {\small $\lambda = $ \gls*{lp} (\ref{eq:reach}) for $\gls*{run}_{seg}$;}
        \STATE {\small \algorithmicif \hspace{1pt}\gls*{lp} (\ref{eq:reach}) is infeasible and $\ell_{i+1}^\prime \geq \gls*{statecnb}$\hspace{1pt}\algorithmicthen \hspace{1pt} $\gls*{run}_{cex}^\prime := \gls*{run}_{seg}$; \hspace{1pt}\algorithmicreturn}
        \STATE {\small \algorithmicelse \hspace{1pt} \algorithmicif \hspace{1pt}\gls*{lp} (\ref{eq:reach}) is infeasible\hspace{1pt}\algorithmicthen \hspace{1pt}\algorithmicreturn}
        \STATE {\small \algorithmicelse \hspace{1pt} \algorithmicif \hspace{1pt}$\lambda \geq \lambda_{\ref{eq:reach}}^{\ast} - \delta$ and $P = P^\ast$\hspace{1pt}\algorithmicthen \hspace{1pt}$\gls*{run}_{cex}^\prime := \gls*{run}_{seg}$; \hspace{1pt}\algorithmicreturn}
        \STATE {\small \algorithmicelse \hspace{1pt} \algorithmicif \hspace{1pt}$\lambda > \delta$ and $P < P^\ast$\hspace{1pt}\algorithmicthen \hspace{1pt}$\gls*{run}_{cex}^\prime := \gls*{run}_{seg}$; \hspace{1pt}\algorithmicreturn}
        \STATE {\small \algorithmicelse \hspace{1pt} \algorithmicif \hspace{1pt}$\lambda > \delta$\hspace{1pt}\algorithmicthen \hspace{1pt}$\lambda_{\ref{eq:reach}} := \lambda$; \hspace{1pt}\algorithmicreturn}
        \STATE {\small $\lambda = $ \gls*{lp} (\ref{eq:feas}) for $\gls*{run}_{seg}$;}
        \STATE {\small \algorithmicif \hspace{1pt}$\lambda \geq \lambda_{\ref{eq:feas}}^{\ast} - \delta$ and $P = P^\ast$\hspace{1pt}\algorithmicthen \hspace{1pt}$\gls*{run}_{cex}^\prime := \gls*{run}_{seg}$; \hspace{1pt}\algorithmicreturn}
        \STATE {\small \algorithmicelse \hspace{1pt} \algorithmicif \hspace{1pt}$\lambda > \delta$ and $P < P^\ast$\hspace{1pt}\algorithmicthen \hspace{1pt}$\gls*{run}_{cex}^\prime := \gls*{run}_{seg}$; \hspace{1pt}\algorithmicreturn}
        \STATE {\small \algorithmicelse \hspace{1pt} \algorithmicif \hspace{1pt}$\lambda > \delta$\hspace{1pt}\algorithmicthen \hspace{1pt}$\lambda_{\ref{eq:feas}} := \lambda$; \hspace{1pt}\algorithmicreturn}
        \STATE {\footnotesize \algorithmicif \hspace{1pt} $i < K^\prime$ \algorithmicthen \hspace{1pt} $\lambda = $ \gls*{lp} (\ref{eq:feasprefix}) for $Prefix(\gls*{run}_p^\prime,P)$;}
        \STATE {\footnotesize \algorithmicelse \hspace{1pt} $\lambda = $ \gls*{lp} (\ref{eq:feasprefix}) for $Prefix(\gls*{run}_p^\prime,P)$ s.t. $\gls*{statec}_H = \gls*{statec}_{N-1}$;}
        \STATE {\small \algorithmicif \hspace{1pt}$\lambda \geq \lambda_{\ref{eq:feasprefix}}^{\ast} - \delta$ and $P = P^\ast$\hspace{1pt}\algorithmicthen \hspace{1pt}$\gls*{run}_{cex}^\prime := Prefix(\gls*{run}_p^\prime,P)$; \hspace{1pt}\algorithmicreturn}
        \STATE {\small \algorithmicelse \hspace{1pt} \algorithmicif \hspace{1pt}$\lambda > \delta$ and $P < P^\ast$\hspace{1pt}\algorithmicthen \hspace{1pt}$\gls*{run}_{cex}^\prime := Prefix(\gls*{run}_p^\prime,P)$; \hspace{1pt}\algorithmicreturn}
        \STATE {\small \algorithmicelse \hspace{1pt} \algorithmicif \hspace{1pt}$\lambda > \delta$\hspace{1pt}\algorithmicthen \hspace{1pt}$\lambda_{\ref{eq:feasprefix}} := \lambda$; \hspace{1pt}\algorithmicreturn}
    \ENDFOR
    \end{algorithmic}
\end{algorithm}

\added{We discuss this algorithm in the following proposition.}

\begin{proposition}\label{prop:search}
    Given a discrete plan $M_p$ and a dynamical system $M_{c}$, Algorithm \ref{alg:feascheck} only returns counter-examples that are \added{constraints of the} \gls*{iis} for a dynamically infeasible discrete plan $M_p$.
\end{proposition}
\begin{proof}
    First, from \cite[Theorem 5.25]{antsaklis2007linear}, an unconstrained discrete-time linear control system is reachable if and only if \gls*{lp} (\ref{eq:reach}) for a segment $s_b^p (s_a^p)^{\ell_a^\prime} s_e^p$ is feasible for $\ell_a^\prime \leq \gls*{statecnb}$ \added{(line 5)}. Moreover, the reachable states of those unconstrainted systems does not change by $\ell_a^\prime > \gls*{statecnb}$. Therefore, if increasing the length $\ell_a^\prime$ the solution of \gls*{lp} (\ref{eq:reach}), the solution does not reduce, the segment $s_b^p (s_a^p)^{\ell_a^\prime} s_e^p$ is not reachable \added{(lines 7 and 8)}. For this reason, from Proposition \ref{prop:nscond}, it is a \added{constraints of the} \gls*{iis}. 
    Now, note that \gls*{lp} (\ref{eq:feas}) for the segment $s_b^p (s_a^p)^{\ell_a^\prime} s_e^p$ is never infeasible if this segment is reachable. Second, if we increase $\ell_a^\prime$, it increases the degree of freedom in the state space. Thus, it must reduce the solution of \gls*{lp} (\ref{eq:feas}) up to its minimum. Therefore, again, if increasing the length $\ell_a^\prime$ the solution of \gls*{lp} (\ref{eq:feas}) does not reduce \added{(lines 11 and 12)}, the segment $s_b^p (s_a^p)^{\ell_a^\prime} s_e^p$ is not feasible and is a \added{constraint of the} \gls*{iis}. 
    Finally, we have a feasible prefix $Prefix(\gls*{run}_d^\prime,P-1)$ and add a feasible segment $s_b^p (s_a^p)^{\ell_a^\prime} s_e^p$ at the end when solving \gls*{lp} (\ref{eq:feasprefix}), where $b = P-1$. Thus, increasing $\ell_i^\prime$ for any $i = 1, \dots, P$ also increases the degree of freedom in the state space. Hence, any increment in an elements of the parameters $\ell_0,\ell_1,\dots,\ell_{P-1}$ reduces \gls*{lp} (\ref{eq:feasprefix}) up to its minimum. Therefore, it does not reduce, the discrete plan $M_p$ is infeasible and the prefix $Prefix(\gls*{run}_d^\prime,P)$ a \added{constraint of the} \gls*{iis} \added{(lines 16 and 17)}, which concludes our theorem.
\end{proof}

\begin{example}\label{ex:feas:cex}
    Consider the running example (\ref{ex:running}) with initial condition $\gls*{statecini} = [-4,-8]^\intercal$. Let us consider the segment $s_{13}(s_{13})^{\ell_0^\prime-1}s_2$ of a run of the discrete plan $M_p$ generated by the run $\gls*{run}_p = s_{13}(s_2s_9s_{11}s_5s_4s_6s_{14}s_{13})^\omega$. The solution of \gls*{lp} (\ref{eq:reach}) for $\ell_0^\prime \in \{1,2,3,4,5,6,7,8 \}$ is $2, 3, 2.33, 1.5, 0.9, 0.47, 0.18, -0.04$, respectively. In fact, after $\ell_0^\prime - 1 > \gls*{statecnb} = 2$, the solution value reduces until it reaches a negative value. In other words, the polytope $\gls*{poly}(s_2)$ is reachable from the initial state. However, if we solve \gls*{lp} (\ref{eq:feas}) for the same segment, we have $1, 5, \dots, 5$, respectively, for the same values of $\ell_0^\prime$. The reason that the solution does not change the value after $\ell_0^\prime = 3$ is that the bounded input turns the prefix $s_{13}(s_{13})^2s_2$ infeasible. Specifically, with input $|u| \leq 2$, the instant $k=2$ inevitably goes outside $s_{13}$ and $s_2$. Therefore, we discard any trajectory with the prefix $s_{13}(s_{13})^*s_2$.
\end{example}

\subsection{Planning Algorithm}

Now, we present how we use Algorithm \ref{alg:feascheck} to search for a dynamically feasible node in the discrete plan $M_p$. In Algorithm \ref{alg:feas}, we propose a search strategy which uses these necessary and sufficient conditions to find a feasible run $\gls*{run}_c \in M_p$. The basic idea is to check that a finite length feasible run exists by checking each \added{segment} and prefix. Each time we check reachability or feasibility, we solve a \gls*{lp} problem as defined in Definitions \ref{def:reach}, \ref{def:feas} and \ref{def:feasprefix}. When such a run cannot be found, this algorithm returns a counterexample representing infeasible prefixes or \added{segments} by computing an \gls*{iis} constraint as per Proposition \ref{prop:search}.

\begin{algorithm}
    \caption{Continuous Motion Planning}\label{alg:feas}
    \begin{algorithmic}
    \REQUIRE {\small $M_{c},M_p,\delta$}
    \ENSURE {\small $cexSet, \gls*{run}_c^*$}
    \STATE \added{{\small $\langle \lambda_{\ref{eq:reach}}, \lambda_{\ref{eq:feas}}, \lambda_{\ref{eq:feasprefix}}, P, \gls*{run}_{cex}^\prime \rangle \gets$ Alg. \ref{alg:feascheck}$(M_p.root,M_c,\delta,\infty,\infty,\infty,0)$;}}
    \STATE {\small $\gls*{run}_c^* := \emptyset$; $cexSet := \emptyset$;$openSet := \langle M_p.root, \lambda_{\ref{eq:reach}}, \lambda_{\ref{eq:feas}}, \lambda_{\ref{eq:feasprefix}}, P, \gls*{run}_{cex}^\prime \rangle$;}
    \WHILE{\small $openSet \neq \emptyset$}
        \STATE {\small $parent := $ {\footnotesize the lowest} $\min(\lambda_{\ref{eq:reach}}, \lambda_{\ref{eq:feas}}, \lambda_{\ref{eq:feasprefix}})$ {\footnotesize with highest} $P$ {\footnotesize from} $openSet$;}
        \STATE {\small remove $parent$ from $openSet$;}
        \STATE \algorithmicif \hspace{1pt} {\small $0 \leq \max(\lambda_{\ref{eq:reach}}, \lambda_{\ref{eq:feas}}, \lambda_{\ref{eq:feasprefix}}) \leq \gls*{tolfeas}$ \algorithmicthen \hspace{1pt} \textbf{break};}
        \FOR {\small $i = 0$ \TO $K^\prime$} 
            \STATE \added{{\small $\gls*{run}_p = repeat\_at(i,parent)$;}}
            \STATE \added{{\footnotesize $\langle \lambda_{\ref{eq:reach}}^\prime, \lambda_{\ref{eq:feas}}^\prime, \lambda_{\ref{eq:feasprefix}}^\prime, P^\prime, \gls*{run}_{cex}^\prime \rangle \gets$ Alg. \ref{alg:feascheck}$(\gls*{run}_p,M_c,\delta,\lambda_{\ref{eq:reach}},\lambda_{\ref{eq:feas}}, \lambda_{\ref{eq:feasprefix}}, P)$;}}
            \STATE {\small $child := \langle \gls*{run}_p, \lambda_{\ref{eq:reach}}^\prime, \lambda_{\ref{eq:feas}}^\prime, \lambda_{\ref{eq:feasprefix}}^\prime, P^\prime, \gls*{run}_{cex}^\prime \rangle$;}
            \STATE \algorithmicif \hspace{1pt} {\small $child.\gls*{run}_{cex}^\prime \neq \emptyset$} \algorithmicthen \hspace{1pt} {\small $cexSet \gets \gls*{run}_{cex}^\prime$} \algorithmicelse \hspace{1pt} {\small $openSet \gets child$;}
        \ENDFOR
        \FOR {\small $i = 1$ \TO $K^\prime$} 
            \STATE \added{{\small $\gls*{run}_p = backward\_at(i,parent)$;}}
            \STATE \added{{\footnotesize $\langle \lambda_{\ref{eq:reach}}^\prime, \lambda_{\ref{eq:feas}}^\prime, \lambda_{\ref{eq:feasprefix}}^\prime, P^\prime, \gls*{run}_{cex}^\prime \rangle \gets$ Alg. \ref{alg:feascheck}$(\gls*{run}_p,M_c,\delta,\lambda_{\ref{eq:reach}},\lambda_{\ref{eq:feas}}, \lambda_{\ref{eq:feasprefix}}, P)$;}}
            \STATE {\small $child := \langle \gls*{run}_p, \lambda_{\ref{eq:reach}}^\prime, \lambda_{\ref{eq:feas}}^\prime, \lambda_{\ref{eq:feasprefix}}^\prime, P^\prime, \gls*{run}_{cex}^\prime \rangle$;}
            \STATE \algorithmicif \hspace{1pt} {\small $child.\gls*{run}_{cex}^\prime \neq \emptyset$} \algorithmicthen \hspace{1pt} {\small $cexSet \gets \gls*{run}_{cex}^\prime$} \algorithmicelse \hspace{1pt} {\small $openSet \gets child$;}
        \ENDFOR
    \ENDWHILE
    \IF { $openSet \neq \emptyset$}
        \STATE {\small find $\gls*{run}_c^*$ solving \gls*{lp} (\ref{eq:feasprefix}) for $parent$ s.t. $\gls*{statec}_H = \gls*{statec}_{N-1}$;}
    \ENDIF
    \end{algorithmic}
\end{algorithm}

\begin{remark}
The root of the discrete plan $M_p.root$ is the shortest run that can be generated from this structure.
\end{remark}

%In fact, a prefix is actually an \gls*{iis} for the feasibility problem defined in Definition \ref{def:feasiblerun} (to see this, consider removing the dynamic constraints $\gls*{inputc}_k \in \gls*{inputcdom}$ and $\gls*{statec}_{k+1} \in Post_{\gls*{tolfeas}}(\gls*{statec}_k,\gls*{inputc}_k)$ or the logic specification constraint $\gls*{statec}_{k+1} \in \gls*{poly}_{k+1}$ for any instant $k$. This renders Problem (\ref{eq:feasiblerun}) feasible). Therefore, we discard any discrete plan that generates these prefixes in future plans. 

This algorithm is sound and complete, as shown by the following proposition:  

\begin{proposition}\label{prop:feas}
    Given a model $M_c$ and a plan $M_p$, Algorithm \ref{alg:feas} returns a run $\gls*{run}_c^*$ if and only if this is a feasible run of $M_p$.  
\end{proposition}
\begin{proof}
    ($\Rightarrow$) It can be easily seen that Algorithm \ref{alg:feas} checks the conditions defined in Proposition \ref{prop:nscond}; thus, any run generated from this algorithm is a feasible run $\gls*{run}_c^*$ of $M_p$. ($\Leftarrow$) Algorithm \ref{alg:feas} considers all possible prefixes of $M_p$. Thus, from Proposition \ref{prop:nscond}, there exists a feasible run $\gls*{run}_c^*$ of $M_p$ only if Algorithm \ref{alg:feas} returns a run $\gls*{run}_c^*$.
\end{proof}

\begin{example} 
Returning to Example \ref{ex:feas:cex}, the run $\gls*{run}_d^\prime = s_{13}(s_2s_9s_{11}s_5s_4s_6s_{14}s_{13})^\omega$ labels the shortest run of the discrete plan $M_p$. Executing Algorithm \ref{alg:feas}, we first obtain the values $\langle \lambda_{\ref{eq:reach}}, \lambda_{\ref{eq:feas}}, \lambda_{\ref{eq:feasprefix}}, P, \gls*{run}_{cex}^\prime \rangle$ for the root, which is $\langle \infty, \infty, \infty, 1, \emptyset \rangle$ because the segment $s_{13}s_2$ is not reachable and $\ell_0^\prime < \gls*{statecnb}$. Next, inside the while loop, we check all the root children. All of them, except the child $repeat\_at(0,parent)$, will return the same values. This child is different because $\ell_0^\prime = 2 = \gls*{statecnb}$; thus, its values are $\langle 3, 5, \infty, 1, \emptyset \rangle$. As a result, this is the next parent and its child generated by $repeat\_at(0,parent)$ and will have values $\langle 2.33, 5, \infty, 1, s_{13}(s_{13})^2s_2 \rangle$ because this node is reachable but not feasible. Since $\gls*{run}_{cex}^\prime$ is not empty and all other children has values $\lambda_{\ref{eq:reach}}$ and $\lambda_{\ref{eq:feas}}$ greater or equal than $2.33$ and $5$, the prefix $s_{13}(s_{13})^2s_2$ is our counter-example.
\end{example}

\section{Iterative Deepening Temporal Logic over Reals}\label{sec:idrtl}

In the previous sections, we presented how to get a discrete plan and check if it is feasible or return a counter-example. Now, we show how to combine both discrete and continuous planning to generate a run of the system (\ref{eq:system}) that ensures the specification. We call our strategy of combining discrete planning and continuous motion planning \emph{Iterative Deepening Temporal Logic over Reals} (\textsc{idRTL}). Algorithm \ref{alg:idrtl} describes this strategy. First, we check if there exists a satisfying solution in the discrete task planning layer. If such a solution exists, we check (continuous motion planning) if there is a corresponding feasible run $\gls*{run}_c$. If so, this is a solution for Problem \ref{prob:1}. 

Otherwise, the continuous planner returns a counter-example and we search for a new discrete plan. This search stops when the formula $\encform{M_d,\gls*{rtlformula},M_{cex},K}_C \gls*{and} \encform{SR}_K$ is unsatisfiable. We iteratively increase the length of the discrete plan $K$, which is the reason we call this algorithm \textsc{idRTL}. If we reach this stop condition, the algorithm returns \textit{no solution}. 

\begin{algorithm}[H]
    \caption{Iterative Deepening Temporal Logic over Reals}\label{alg:idrtl}
    \begin{algorithmic}
    \REQUIRE {\small $M_{c},\gls*{rtlformula},\delta$}
    \STATE {\small $K \gets 0$;}
    \WHILE{\small $\encform{M_d,\gls*{rtlformula},M_{cex},K}_C \gls*{and} \encform{SR}_K$ is \textbf{satisfiable}}
        \IF {\small $\encform{M_d,\gls*{rtlformula},M_{cex},K}$ is \textbf{satisfiable}}
            \STATE  \algorithmicif \hspace{1pt} {\small continuous motion planning returns a counter-example $M_{cex}^i$} \algorithmicthen \hspace{1pt} {\small $M_{cex} \gets M_{cex} \cup M_{cex}^i$}
            \STATE \algorithmicelse \hspace{1pt} \algorithmicreturn \hspace{1pt} {\small  ($\gls*{run}_p$, $\gls*{run}^*$);}
        \ELSE
            \STATE {\small $K \gets K + 1$;}
        \ENDIF
    \ENDWHILE
    \STATE \algorithmicreturn \hspace{1pt} {\small  No solution for Problem \ref{prob:1} }
    \end{algorithmic}
\end{algorithm}

\subsection{Soundness}

We show that Algorithm \ref{alg:idrtl} is sound in the sense that any run $\gls*{run}^*$ generated by Algorithm \ref{alg:idrtl} solves Problem \ref{prob:1}.  

\begin{theorem}\label{theo:mainsoundness}
    Given an \gls*{rtl} formula \gls*{rtlformula} and a dynamical system $M_{c}$, any solution $\gls*{run}_c^*$ of Algorithm \ref{alg:idrtl} is a solution for Problem \ref{prob:1}.
\end{theorem}
\begin{proof}
    From Proposition \ref{prop:dplan_soundness}, the discrete plan $M_p$ generated in the discrete task planning phase enforces satisfaction of the specification $\varphi$. From Proposition \ref{prop:feas}, the solution $\gls*{run}_c^*$ generated in the continuous motion planning phase is a feasible run of a $M_p$. Therefore, the continuous run $\gls*{run}_c^*$ is a solution for Problem \ref{prob:1}.
\end{proof}

\subsection{Completeness}

We show that Algorithm \ref{alg:idrtl} is complete in the sense that if no solution is found, then no solution to Problem \ref{prob:1} exists. 

\begin{theorem}\label{theo:maincompleteness}
    Given an \gls*{rtl} formula \gls*{rtlformula} and a dynamical system $M_{c}$, Algorithm \ref{alg:idrtl} returns no solution only if there exists no solution for Problem \ref{prob:1}.
\end{theorem}
\begin{proof}
    As discussed in Section \ref{sec:fsearch}, we assume that a solution must be finite or periodic. As a result, from Proposition \ref{prop:feas}, if Algorithm \ref{alg:feas} returns a counter-example for a discrete plan $M_p$, then there exists no feasible run for $M_p$. Since the counter-example is a \gls*{iis}, next $M_p$ is always different. From Proposition \ref{prop:dplan_soundness} and Proposition \ref{prop:dplan_completeness}, if for some $K\geq 0$ $\encform{M_d,\gls*{rtlformula},M_{cex},K}_C \gls*{and} \encform{SR}_K$ is unsatisfiable and either $K = 0$ or $\encform{M_d,\gls*{rtlformula},M_{cex},K-1}$ is unsatisfiable, then there exists no discrete plan $M_p$ that satisfies the specification. From Proposition \ref{prop:simulation}, it then follows that there exists no solution for Problem \ref{prob:1}.
\end{proof}

\subsection{Complexity}

The complexity of \textsc{idRTL} depends on the \gls*{rtl} formula and the number of continuous variables in $M_c$. First, the worst case number of discrete states in the abstraction $M_d$ is exponential in number of atomic predicates, $O(2^{|\gls*{predset}|})$. Our proposed encoding $\encform{M_d,\gls*{rtlformula},M_{cex},K}$ depends linearly on the length $K$ and number of subformulas i.e., $O(K|cl(\gls*{rtlformula})|)$, and quadraticaly on $\encform{SR}_K$, i.e., $O((K|cl(\gls*{rtlformula})|)^2)$. Finally, the complexity of \gls*{lp} for continuous motion planning depends linearly on the length $H$ of a continuous run $\gls*{run}_c$ and linearly on the number of continuous variables of $M_c$, i.e., $O(H(\gls*{statecnb}+\gls*{inputcnb})+\gls*{statecnb})$ for the number of variables and $O(H\gls*{statecnb})$ for the number of constraints. 

\begin{remark}
Note that the complexity does not directly depend on the complexity of the atomic predicates. As a result, we can achieve the same performance for arbitrary polytopic constraints as for simpler (rectangular) constraints.
\end{remark}

\section{Simulation}\label{sec:results}

In this section, we demonstrate the scalability of our approach and compare with state-of-the-art solvers in three scenarios. First, we show that our approach quickly determines whether a given initial state has a dynamically feasible satisfying trajectory. We contrast these results with LanGuiCS solver \cite{belta2017formal} which can also determine the inexistence of a solution given an initial state.  Next, we evaluate the performance of our approach on a motion planning problem, where the algorithm must be scalable to long trajectories and non-convex constraints. We compare these results with the OMPL solver \cite{sucan2012the-open-motion-planning-library} (sampling-based motion planning), and LTLOpt solver \cite{wolff2016optimal} (\gls*{milp} motion planning). Finally, we validate the scalability of our approach to high-dimensional dynamical systems by considering a quadrotor model with 18 continuous variables. We also compare these results with the SatEX solver \cite{shoukry2018smc}, whose main focus is scalability to high-dimensional systems. 

We implemented our approach using Z3 SMT solver \cite{de2008z3}, Gurobi LP solver \cite{gurobi}, and \textit{lrs} vertex enumeration solver \cite{avis2013portable} and is available at \url{https://bitbucket.org/rafael_rodriguesdasilva/idrtl/}. All experiments were executed on an Intel Core i7 processor with 32GB RAM.

\subsection{Determining Existence of a Satisfying Trajectory}

In this experiment, we use \cite[Example 11.5]{belta2017formal} as a benchmark problem. This problems uses the same system but with different specification, i.e., $\gls*{rtlformula} = (\idstluntil{{\gls*{neg}} a}{}{b}) \gls*{and} (\idstluntil{\neg b}{}{c})$, where $a$ is represented by the blue regions in Figure \ref{fig:completeness}, $b$ corresponds to the red region, and $c$ to the yellow regions.
This problem is a particularly challenging problem \cite{gol2013language}. The forbidden regions create areas in the workspace where no solution exists. Thus, it is particularly important to decide when the dynamical constraints render a given initial state infeasible.  

In this scenario, we selected $10$ feasible and $10$ unfeasible initial states and executed the idRTL and LanGuiCS solvers. Fig.~\ref{fig:completeness-idrtl} illustrates the solutions generated by idRTL for these initial states. Black stars are initial states for which Algorithm \ref{alg:idrtl} returned \textit{no solution}. The curves are trajectories for initial states where a solution exists. idRTL took $52.6 \pm 5.9ms$ to compute a solution when a solution exists, and $156.9 \pm 96.1ms$ to determine that no solution exists. 

Fig.~\ref{fig:completeness-languics} shows the same results for LanGuiCS. We can see that both algorithms correctly determined the existence of a solution. LanGuiCS took $1352.6 \pm 784.0ms$ to generate solutions for feasible initial states and $22.3 \pm 1.5ms$ to decide that no solution exists. This solver can achieve a quicker decision for unfeasible initial states because it computes the feasible regions offline. However, LanGuiCS required more than $10$min ($612.606$s) to compute these regions. 

These results show that idRTL makes decisions on the existence of a solution in a reasonable time without offline computation. Furthermore, the scalability of our approach is even more evident when a solution exists. The reason for this is that it only takes one valid discrete plan to decide that a solution exists, but it can take several of these plans to rule out the existence of a feasible solution.   

\begin{figure}[htp]
\centering
\subfloat[idRTL]{\label{fig:completeness-idrtl}
\includegraphics[width=0.245\textwidth]{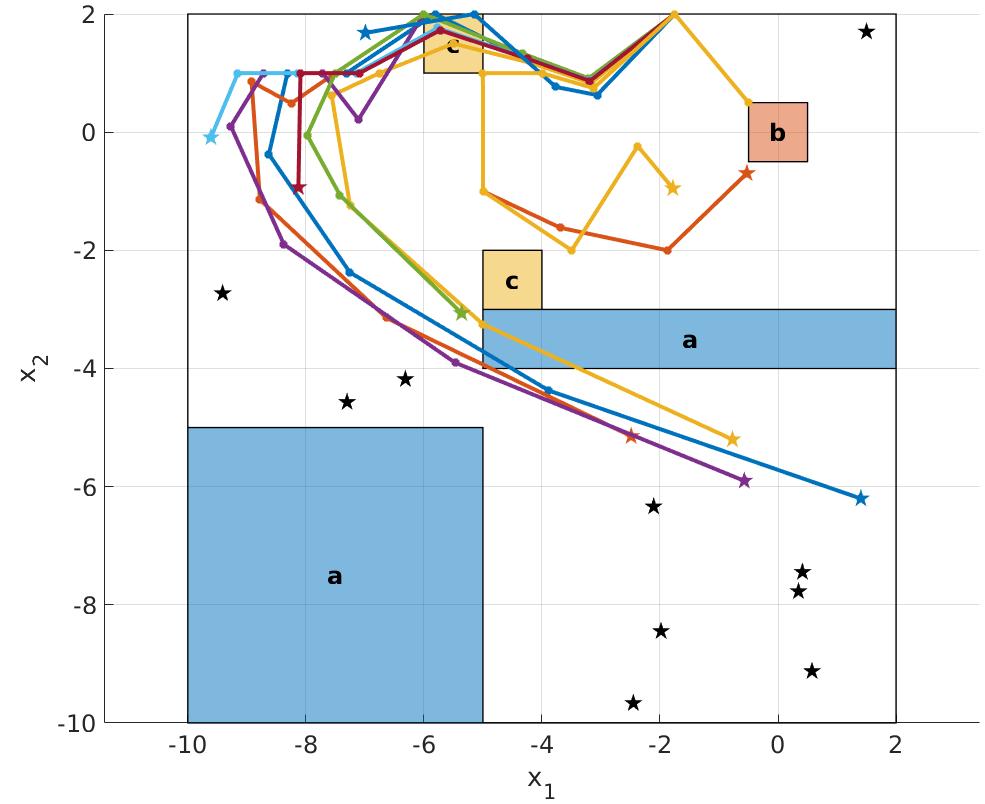}}
\subfloat[LanGuiCS]{\label{fig:completeness-languics}
\includegraphics[width=0.245\textwidth]{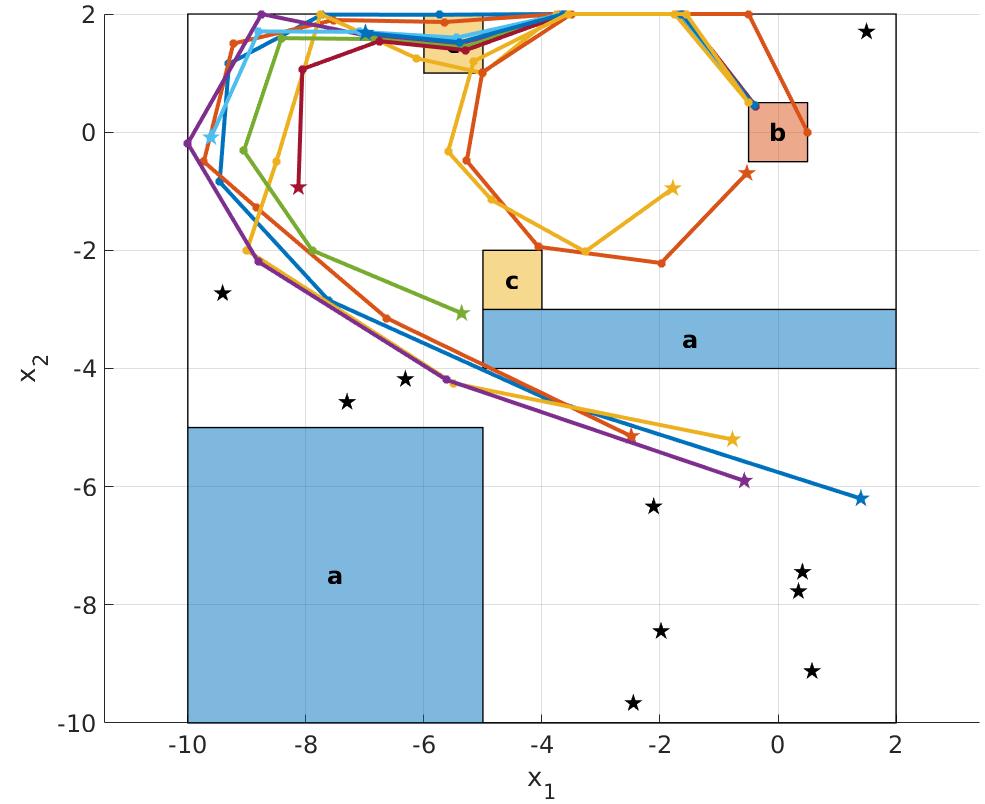}}
\caption{Comparation between idRTL and LanGuiCS. \added{Note that some continuous trajectories appear to pass through the forbidden blue region, which is an artifact of the time discretization.  In practice, we can avoid it by increasing the obstacles' size or changing the discretization strategy.}}
\label{fig:completeness}
\end{figure}

\subsection{Application to Motion Planning}

We consider a motion planning problem scenario where a mobile robot must reach a target region while avoiding collisions with obstacles. The main challenge is in considering narrow passages. As the number of passages grows, so does the size of the discrete abstraction.

We model the robot as a two-dimensional double integrator with sampling time $T_s = 0.5s$, where $\gls*{statec}_k \in \gls*{statecdom} \subseteq \mathbb{R}^4$, $\gls*{inputc}_k \in \gls*{inputcdom} \subseteq \mathbb{R}^2$, $\gls*{statecdom} = \{ \gls*{statec} \in \mathbb{R}^4 : 0 \leq x_1 \leq 30, 0 \leq x_2 \leq 30, -2 \leq x_3 \leq 2, -2 \leq x_4 \leq 2 \}$, $\gls*{inputcdom} = \{ \gls*{inputc} \in \mathbb{R}^2 : \| \gls*{inputc} \|_{\infty} \leq 0.5 \}$. The workspace is illustrated in Fig.~\ref{fig:mp_runs}, where grey regions indicate obstacles and green is the target region. 

\begin{figure}[htp]
\centering
\subfloat[One passage.]{\label{fig:mp_runs1}
\includegraphics[width=0.22\textwidth]{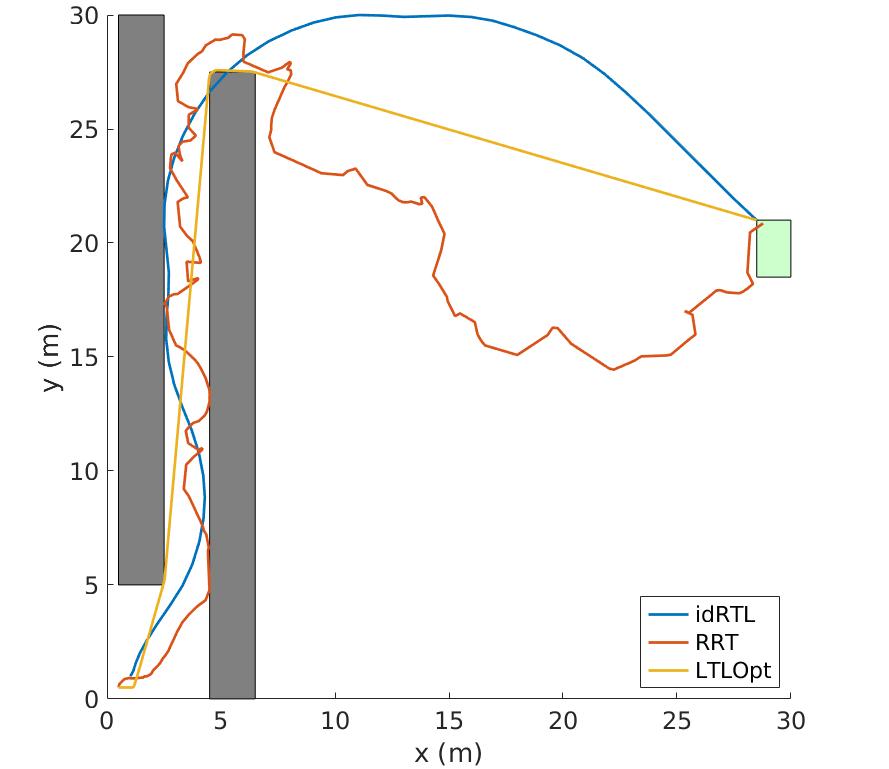}}
\subfloat[Four passages.]{\label{fig:mp_runs4}
\includegraphics[width=0.22\textwidth]{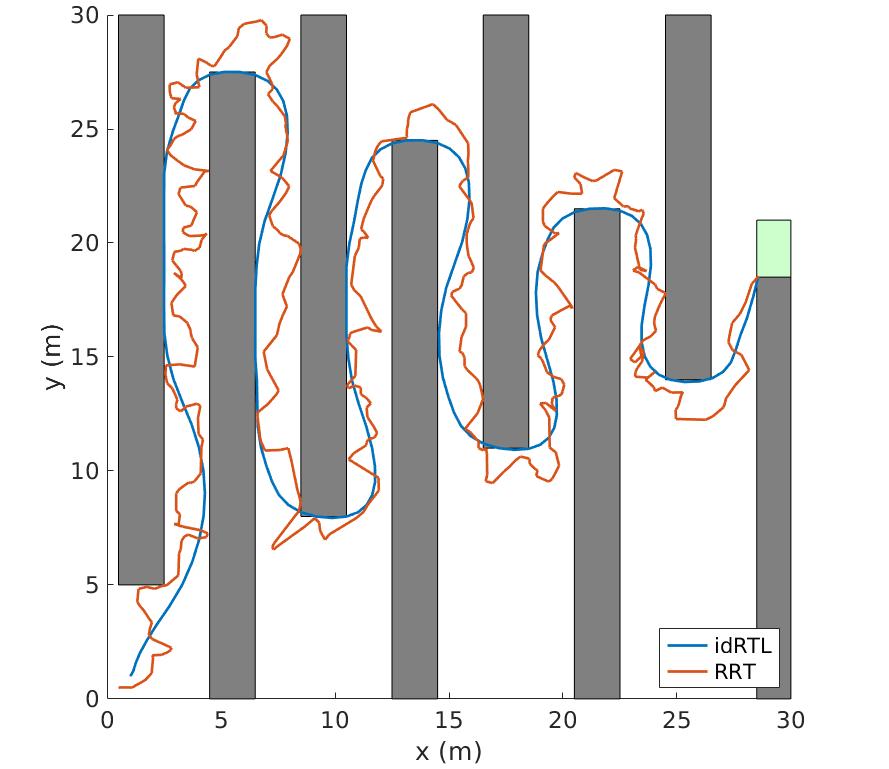}}
\caption{Comparison with sampling-based and MILP-based solvers on a maze problem.}
\label{fig:mp_runs}
\end{figure}

Table \ref{tab:compare} and Fig.~\ref{fig:mp_runs} demonstrate the scalability of idRTL in this motion planning application. In Table \ref{tab:compare}, we present the computation performance of idRTL and compare with state-of-art sampling-based and \gls*{milp}-based (LTLOpt \cite{wolff2016optimal}) approaches. This table shows the average run-time of $10$ executions for different numbers of passages. The number of passages increases the number of obstacles and the length of a satisfying run, as shown in Fig.~\ref{fig:mp_runs1} and \ref{fig:mp_runs4}. idRTL is consistently faster than the other approaches. This demonstrates that considering logical and dynamical constraints with a comination of SAT/SMT and optimization solvers is more efficient than solving a \gls*{milp} problem. The encoding of logical constraints in a \gls*{milp} problem is especially costly for longer runs. We also compared our approach with RRT  \cite{Lavalle98rapidly-exploringrandom}, a probabilistically complete motion planning algorithm \cite{karaman2011sampling}. Even though RRT solves a much less expressive problem, idRTL can sill be be an order of magnitude faster.

\begin{table}[] 
\caption[Run-time comparation between idRTL and different state-of-art approaches for a maze-like workspace.]{Run-time comparation between idRTL and different state-of-art approaches for a maze-like workspace. }
\centering
\begin{tabular}{c|c|c|cc}
                                      &            &                           & \multicolumn{2}{c}{LTLOpt (s)}   \\
\multirow{-3}{1.4cm}{Number of passages} &  \multirow{-2}{*}{idRTL(s)} & \multirow{-2}{*}{RRT (s)} & Solver         & Yalmip        \\ \hline
 1                                    & \cellcolor[HTML]{9AFF99}0.360    & 3.504                     & 146.93         & 1.536     \\
 2                                    & \cellcolor[HTML]{9AFF99}1.415    & 9.574                     & timeout            & timeout     \\
 3                                    & \cellcolor[HTML]{9AFF99}3.488    & 16.601                    & timeout            & timeout       \\
 4                                    & \cellcolor[HTML]{9AFF99}6.611    & 20.956                    & timeout            & timeout        \\
\end{tabular}
\label{tab:compare}
\end{table}

\subsection{High-dimensional dynamical systems}

We demonstrate the scalability of our approach to high-dimensional systems with an example of motion planning for a quadrotor. The quadrotor moves in 3-dimensional Euclidean space and operates with linearized dynamics having $18$ continuous variables. We compare these results with SatEx solver, which is also scalable to high-dimensional systems. We consider the collision avoidance scenario presented in \cite{shoukry2018smc}, and increase the complexity of the problem by increasing the length of the $x$-axis and the number of obstacles. 
 
Fig.\ref{fig:smc} shows that our approach can find satisfying solutions for this high-dimensional system even when we increase the problem complexity. idRTL is consistently faster in computing these runs than SatEx. The main reason for this is that our approach searches first for shorter discrete plans to generate the continuous runs. Since the non-convex nature of the problem is generated by the logical constraints, starting with shorter discrete runs reduces the non-convexity in the problem. 

\begin{figure}[htp]
\centering
\begin{tikzpicture}[thick,scale=0.4, every node/.style={scale=1.8}] 
	\begin{semilogyaxis}[ xlabel={$x$-axis length}, ylabel={Run-time (s)},legend style={at={(0.8,0.05)},anchor=south},xtick={6,7,8,9,10,11,12,13,14,15},width=16cm,height=8cm  ] 
		\addplot coordinates{ (6, 1.334) (7,1.339) (8,1.955) (9,2.410) (10,3.037) (11,3.346) (12,3.932) (13,4.754) (14,5.522) (15,5.956) }; 
		\addplot coordinates{ (6, 5.35142302513) (7,5.6817920208) (8,10.6562080383) (9,11.1644887924) (10,27.4904651642) (11,19.1369888783) (12,49.2101881504) (13,39.9829249382) (14,83.6151809692) (15,45.2980251312) }; 
		\legend{$idRTL$,$SatEx$} 
	\end{semilogyaxis} 
\end{tikzpicture}
\caption{The experiment in \cite[Sec. 6.3]{shoukry2017smc}, for which code is available.
%\footnote{\url{https://yshoukry.bitbucket.io/SatEX/}}
Longer $x$-axis also have more obstacles. }
\label{fig:smc}
\end{figure}
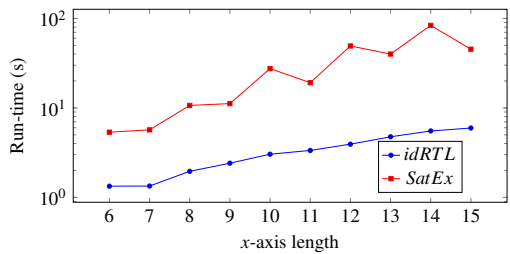

\section{Conclusion}\label{sec:conclusion}

We proposed a fast, scalable, and provably complete symbolic control method for unbounded temporal logic specifications. To address the coupling between nonconvex logical constraints and physical dynamic constraints, we designed a two-layer control architecture which separates discrete task planning and continuous motion planning on-the-fly. By directly addressing this core problem, our approach scales well to high-dimensional systems and complex specifications, as well as offering order-of-magnitude speed improvements over the current state-of-the-art. We hope that this work will provide a step towards safe and provably correct control of complex autonomous \gls*{ips}s. Future work will focus on extensions to unknown/dynamic environments and non-linear/hybrid systems. 

\bibliographystyle{IEEEtran}
\bibliography{IEEEabrv,library}

\begin{IEEEbiography}
    [{\includegraphics[width=1in,height=1.25in,clip,keepaspectratio]{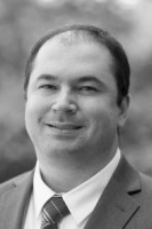}}]{Rafael Rodrigues da Silva} (M'16) received a bachelor in science (BS) degree in Electrical Engineering from the State University of Santa Catarina, in Joinville, Brazil, his hometown. Before going into academia, he had worked for almost ten years in the steel industry as a control engineer, where he lead a team of high qualified engineers and technicians in several projects. During this time, he received a master in science (MS) degree in Electrical Engineering and Industrial Computing at UTFPR, in Curitiba, Brazil, when he focused his research in high-performance Genetic Algorithms in FPGA (Field Programmable Gate Array) for Computer Vision. He is currently pursuing the Ph.D. degree in electrical engineering at University of Notre Dame, Notre Dame, IN, USA.
    
    His research interest focus on the design of intelligent physical systems, and, currently, he is working with verification and automatic synthesis of hybrid systems from high-level specifications.
\end{IEEEbiography}

\begin{IEEEbiography}
    [{\includegraphics[width=1in,height=1.25in,clip,keepaspectratio]{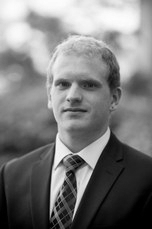}}]{Vince Kurtz}
Vince Kurtz studied physics at Goshen College (Goshen Indiana, 46526) and is currently a PhD student in electrical engineering at the University of Notre Dame (Notre Dame Indiana, 46556). His research interest in long-term autonomy for uncertain systems lies at the intersection of robotics, control theory, and formal methods. 
\end{IEEEbiography}

\begin{IEEEbiography}
    [{\includegraphics[width=1in,height=1.25in,clip,keepaspectratio]{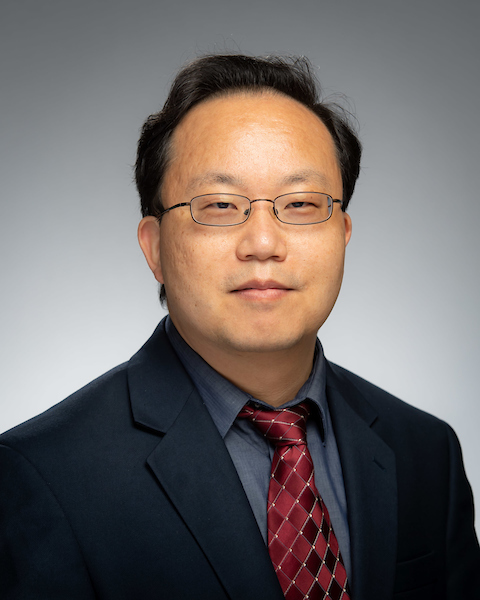}}]{Hai Lin} 
(SM’10) is currently a full professor at the Department of Electrical Engineering, University of Notre Dame, where he got his Ph.D. in 2005. Before returning to his {\em alma mater}, he has been working as an assistant professor in the National University of Singapore from 2006 to 2011. Dr. Lin's recent teaching and research activities focus on the multidisciplinary study of fundamental problems at the intersections of control theory, machine learning and formal methods. His current research thrust is motivated by challenges in cyber-physical systems, long-term autonomy, multi-robot cooperative tasking, and human-machine collaboration. Dr. Lin has been served in several committees and editorial board, including {\it IEEE Transactions on Automatic Control}. He served as the chair for the IEEE CSS Technical Committee on Discrete Event Systems from 2016 to 2018, the program chair for IEEE ICCA 2011, IEEE CIS 2011 and the chair for IEEE Systems, Man and Cybernetics Singapore Chapter for 2009 and 2010. He is a senior member of IEEE and a recipient of 2013 NSF CAREER award.
\end{IEEEbiography}

\end{document}